\documentclass[acmsmall,screen]{acmart}
\AtBeginDocument{%
  }

\setcopyright{rightsretained}

\acmDOI{10.1145/3808348}
\acmYear{2026}
\acmJournal{PACMPL}
\acmVolume{10}
\acmNumber{PLDI}
\acmArticle{270}
\acmMonth{6}

\usepackage{mathtools}
\usepackage{dsfont}
\mathtoolsset{showonlyrefs=true}
\usepackage{stmaryrd}
\usepackage{mathpartir}
\usepackage{extarrows}
\usepackage{multirow}
\usepackage{subcaption}

\newcommand{\ifexpr}[3]{\mathbf{if}\ #1\ \mathbf{then}\ #2\ \mathbf{else}\ #3} %
\newcommand{\ifstmt}[3]{\mathtt{if}\ #1\ \mathtt{then}\ #2\ \mathtt{else}\ #3} %
\newcommand{\while}[2]{\mathtt{while}\ (#1)\ \{\ #2\ \}}
\newcommand{\assign}[2]{#1 \coloneqq #2}
\newcommand{\pbranch}[3]{#1 \ [#2]\ #3}
\newcommand{\skipstmt}{\mathtt{skip}}
\newcommand{\tickstmt}{\mathtt{tick}}
\newcommand{\scoring}[1]{\mathtt{score}(#1)}
\newcommand{\observe}[1]{\mathtt{observe}(#1)}
\newcommand{\interpret}[1]{\llbracket #1 \rrbracket}
\newcommand{\restrictby}[2]{#1_{\le #2}}
\newcommand{\probzero}{\mathds{O}}
\newcommand{\probone}{\mathds{1}}
\newcommand{\Diff}{\mathsf{D}}
\newcommand{\nonterm}[1]{\Diff{#1}}
\newcommand{\argmin}{\mathop{\mathrm{arg~min}}}
\makeatletter
\newcommand{\OurTool}{%
\if@ACM@anonymous
\textsc{Anon}%
\else
$\textsc{MuVal}^{\textsc{QFL}}$%
\fi}
\newcommand{\OurToolUrl}{%
\if@ACM@anonymous
Anonymised for review.
\else
\url{https://github.com/hiroshi-unno/coar}
\fi}
\makeatother

\newif\iflong
\longtrue

\usepackage{refcount}
\newcommand{\checkrefnumber}[2]{%
	\IfRefUndefinedExpandable{#2}{}{%
		\ifthenelse{\equal{#1}{\getrefnumber{#2}}}{}{%
		\PackageError{macros.apxproof}{Reference number mismatch: expected #1 but got \getrefnumber{#2}}{}}}}

\iflong
\newcommand{\referappendix}[3]{%
		\checkrefnumber{#2}{#3}%
		#1~\ref{#3}}
\else
\newcommand{\referappendix}[3]{\cite[#1~#2]{arxiv}}
\fi

\newif\ifdraft
\draftfalse

\ifdraft
\newcommand{\todo}[1]{\textcolor{red}{[ToDo: #1]}}
\newcommand{\tk}[1]{\textcolor{blue}{[{#1}--Takeshi]}}
\else
\newcommand{\todo}[1]{}
\newcommand{\tk}[1]{}
\fi

\AtEndPreamble{%
\theoremstyle{acmdefinition}

\newtheorem{notation}[theorem]{Notation}
\newtheorem{remark}[theorem]{Remark}}

\begin{document}
\title{Supermartingales for Unique Fixed Points: A Unified Approach to Lower Bound Verification}
\author{Satoshi Kura}
\email{satoshikura@acm.org}
\orcid{0000-0002-3954-8255}
\affiliation{%
  \institution{Waseda University}
  \city{Tokyo}
  \country{Japan}
}

\author{Hiroshi Unno}
\email{hiroshi.unno@acm.org}
\orcid{0000-0002-4225-8195}
\affiliation{%
  \institution{Tohoku University}
  \city{Sendai}
  \country{Japan}
}

\author{Takeshi Tsukada}
\email{t.tsukada@acm.org}
\orcid{0000-0002-2824-8708}
\affiliation{%
  \institution{Chiba University}
  \city{Chiba}
  \country{Japan}
}

\begin{abstract}
Many quantitative properties of probabilistic programs can be characterized as least fixed points, but verifying their lower bounds remains a challenging problem.
We present a new approach to lower-bound verification that exploits and extends the connection between the uniqueness of fixed points and program termination.
The core technical tool is a generalization of ranking supermartingales, which serves as witnesses of the uniqueness of fixed points.
Our method provides a simple and unified reasoning principle applicable to a wide range of quantitative properties, including termination probability, the weakest preexpectation, expected runtime, higher moments of runtime, and conditional weakest preexpectation.
We provide a template-based algorithm for automated verification of lower bounds and demonstrate the effectiveness of the proposed method via experiments.
\end{abstract}



\maketitle

\section{Introduction}

Reasoning about probabilistic programs is a challenging problem, and various approaches have been developed.
One of the most fundamental techniques is weakest-precondition-style calculi, which have been extended for various quantitative properties of probabilistic programs.
The weakest preexpectation transformer~\cite{McIver2005} is a natural extension of the weakest precondition transformer~\cite{DijkstraComACM1975} and further extended for the expected runtime~\cite{KaminskiJACM2018}, hard/soft conditioning~\cite{OlmedoTOPLAS2018,SzymczakSETSS2020}, and higher moments of runtime~\cite{KuraTACAS2019,AguirreMSCS2022}.
Similarly to the case for non-probabilistic programs, these calculi characterize the behaviour of a loop as the least fixed point of some function.
Therefore, upper/lower bounding the least fixed point is a key problem in the verification of quantitative properties of probabilistic programs.

\begin{example}
    As an example, we consider the following biased random walk (where \( x \) ranges over \( \mathbb{N} \)) and analyse the probability of termination from state \( x = 1 \):
    \begin{equation*}
        P_1
        \quad\equiv\quad
        \mathtt{while}(x > 0) \big\{\; \mathtt{if}(\mathtt{random\_bool}(1/3))\;\{\,x := x - 1\,\}\;\mathtt{else}\;\{ x := x + 1 \;\}\,\big\}
        \label{eq:intro:biased-random-walk-example}
    \end{equation*}
    where \( \mathtt{random\_bool}(p) \) is a function that returns \( \mathtt{true} \) with probability \( p \) and \( \mathtt{false} \) with probability \( (1-p) \).
    The analysis can be reduced to the computation of the least fixed point of \( K_1 \colon (\mathbb{N} \to [0,1]) \to (\mathbb{N} \to [0,1]) \) defined by 
    \begin{equation}
        K_1(X \colon \mathbb{N} \to [0,1])(x \colon \mathbb{N}) \quad\coloneqq\quad \ifexpr{x > 0}{\tfrac{1}{3} X(x - 1) + \tfrac{2}{3} X(x + 1)}{\probone}
        \label{eq:intro:example}
    \end{equation}
    (where \( \probone \) denotes the constant function with value \( 1 \)).
    Writing \( (\mu K_1) \colon \mathbb{N} \to [0,1] \) for the least fixed point of \( K_1 \), the probability of termination from state \( x = 1 \) is \( (\mu K_1)(1) \), of which the exact value is \( (\mu K_1)(1) = 1/2 \).
    \qed
\end{example}

An upper bound of \( (\mu K_1)(1) \) can be relatively easily obtained by using the \emph{Knaster--Tarski theorem} or \emph{Park induction}.
It is the following reasoning principle:
\begin{equation}
    K(\eta) \le \eta \quad\Longrightarrow\quad \mu K \le \eta
    \label{eq:knaster-tarski}
\end{equation}
where \( \eta \colon \mathbb{N} \to [0,1] \) is an arbitrary assignment.
For example, if we define \( \eta \) by \( \eta(n) := (1/2)^n \), then \( K_1(\eta) \le \eta \) holds.
It follows from the above rule that \( (\mu K_1) \le \eta \), and hence \( (\mu K_1)(1) \le \eta(1) = (1/2) \).

In contrast, obtaining a lower bound is not straightforward, as pointed out by \citet{HarkPOPL2020}.
The following reasoning principle, obtained by reversing the order of the previous argument,
\begin{equation}
    \eta \le K(\eta) \quad\stackrel{\textit{wrong}}{=\!=\!=\!\Longrightarrow}\quad \eta \le \mu K
    \label{eq:wrong-principle}
\end{equation}
is wrong.
For example, letting \( \eta(n) := 1 \) for every \( n \), we have \( \eta \le K_1(\eta) \) but \( \eta(1) = 1 \not\le 1/2 = (\mu K_1)(1) \).
The correct version of the above rule is \( \eta \le K(\eta) \Longrightarrow \eta \le \nu K \), where \( \nu K \) is the greatest fixed point, but this is a lower-bound of the greatest fixed-point, not of the least fixed-point.

The challenge of obtaining a lower bound has been addressed in several papers \cite{HarkPOPL2020,ChatterjeeCAV2022,FengOOPSLA2023,MajumdarPOPL2025,AbateCAV2025,HenzingerCAV2025,WangPLDI2021a,McIver2005,UrabeLICS2017,TakisakaATVA2018}.
One line of work is based on \emph{uniform integrability} and the \emph{optional stopping theorem} (OST)~\cite{HarkPOPL2020,WangPLDI2019,WangPLDI2021a}, which provide remedies to the invalid rule \eqref{eq:wrong-principle} by imposing additional conditions.
Uniform integrability, a notion from probability theory, is known to be closely related to lower bounds of least fixed points~\cite{HarkPOPL2020}.
Since uniform integrability itself is generally difficult to check directly, the OST and its variants are often used to give sufficient conditions for it.
Such sufficient conditions typically impose both termination requirements on programs and certain forms of boundedness, such as \emph{conditional difference boundedness}~\cite{FerrerFioritiPOPL2015,FuVMCAI2019,HarkPOPL2020} and \emph{bounded update}~\cite{WangPLDI2019,WangPLDI2021a}.
Another line of work is based on approximations~\cite{ChatterjeeCAV2022,MajumdarPOPL2025,FengOOPSLA2023,BeutnerPLDI2021}, where lower bounds are obtained by finding subsets of the state space within which the target quantity can be computed exactly.

However, there still remain challenges.
Approximation-based approaches, by their nature, struggle to infer exact lower bounds.
Even when completeness results guarantee arbitrarily precise lower bounds, they typically require countably infinite certificates to obtain the exact lower bound.
On the other hand, OST-based approaches often require case-specific treatment for each type of transformer, making it difficult to uniformly handle the diverse transformers discussed earlier, even though these transformers can be uniformly defined in terms of fixed points.
Moreover, programs that are not almost-surely terminating fall out of the scope of OST-based approaches.
Indeed, even for seemingly simple problems such as the termination probability of \( P_1 \), both of existing approaches have difficulty in inferring exact lower bounds.

\paragraph{Our Approach}

This paper proposes a new method for lower-bound estimation that is uniformly applicable to all transformers mentioned at the beginning and can reason about the exact bound for \( P_1 \).
A key underlying idea of our approach is to exploit the \emph{uniqueness} of fixed points.
It is easy to see that if a function \( K \) has a unique fixed point, then the reasoning principle \eqref{eq:wrong-principle} becomes valid.
\begin{equation}
    \eta \le K(\eta)
    \quad\xLongrightarrow{\!\!\textit{Knaster-Tarski}\!\!}\quad
    \eta \le \nu K
    \quad\xLongrightarrow{\!\!\nu K = \mu K\!\!}\quad
    \eta \le \mu K
    \label{eq:ufp-framework}
\end{equation}
Then, the question is how to ensure the uniqueness of fixed points.
For weakest preexpectation transformers over bounded expectations (e.g.~those taking values in \( [0,1] \)), the necessary and sufficient condition for the existence of a unique fixed point is already known: as observed in a classical paper by Kozen~\cite[Corollary~5.2.5]{Kozen1979}, it is the (almost-sure) termination.
This suggests that lower bounds can be witnessed by the combination of $\eta$ as in \eqref{eq:wrong-principle} and a \emph{ranking supermartingale} \cite{ChakarovCAV2013} witnessing the almost-sure termination; indeed, this approach has been adopted in~\cite{WangPLDI2021}.
However, this reasoning is generally invalid for the unbounded case such as expected runtime \( \mathbb{N} \to [0,\infty] \).
\begin{example}
    Consider, for example, a random walk biased toward \( 0 \):
    \begin{equation*}
        \mathtt{while}(x > 0) \big\{\; \mathtt{if}(\mathtt{random\_bool}(2/3))\;\{\,x := x - 1\,\}\;\mathtt{else}\;\{ x := x + 1 \;\}\,\big\}
    \end{equation*}
    Let us examine the expected runtime, which is expressed as the least fixed point of \( K_2 \colon (\mathbb{N} \to [0,\infty]) \to (\mathbb{N} \to [0,\infty]) \) given by
    \begin{equation*}
        K_2(X \colon \mathbb{N} \to [0,\infty])(x \colon \mathbb{N}) \quad\coloneqq\quad \ifexpr{x > 0}{\tfrac{2}{3} X(x - 1) + \tfrac{1}{3} X(x + 1) + 1}{0}.
    \end{equation*}
    This version of biased random walk is almost-surely terminating, but \( \mu K_2 \neq \nu K_2 \):
    the least fixed point is \( \eta(x) := 3x \) but \( \xi(x) = 3x + (2^x -1) \) is another fixed point.
    \qed
\end{example}
Even in the bounded case, the above idea does not apply to non-terminating programs such as \( P_1 \).
This raises the following question: how can we ensure the uniqueness of fixed points in such cases?

\paragraph{Ranking Supermartingales for Unique Fixed Points}
This paper introduces a new generalization of ranking supermartingales~\cite{ChakarovCAV2013} for ensuring the uniqueness of fixed points.
Our idea is to extend the usual ranking condition for ranking supermartingales, which is of the form \( (\Diff K)(R) + 1 \le R \) in our notation, into a more general condition \( (\Diff K)(R) + U \le R \) relative to a choice of a function $U$.

More formally, we consider an upper bound \( U \) of the least fixed point and aim to prove that \( K \) has a unique fixed point in the subset \( \{ X \colon \mathbb{N} \to [0, \infty] \mid X \le U \} \subseteq (\mathbb{N} \to [0,\infty]) \) dominated by \( U \).
This is achieved by our supermartingale-based method stated as follows: it suffices to find \( R \colon \mathbb{N} \to [0, \infty) \) that satisfies \( (\Diff K)(R) + U \le R \) (where \( (\Diff K)(X) = K(X) - K(\probzero) \) and \( \probzero \) denotes the constant function with value 0), which we call a \emph{\( U \)-ranking supermartingale}.
Using \( U \)-ranking supermartingale, a generalized proof rule for lower bound is given as follows:
\begin{equation}
    K(U) \le U
    \quad\wedge\quad
    (\Diff K)(R) + U \le R
    \quad\wedge\quad
    \eta \le U
    \quad\wedge\quad
    \eta \le K(\eta)
    \qquad\Longrightarrow\qquad
    \eta \le \mu K.
    \label{eq:proof-rule-unbounded}
\end{equation}
The bounded case discussed above is subsumed as a special case where \( U = \probone \) and \( K \) is instantiated as the weakest preexpectation transformer of a program.\footnote{Note that the left-most condition becomes \( K(\probone) \le \probone \), which trivially holds if \( K \) is the weakest precondition transformer of a probabilistic program.}

Interestingly, this proof rule is also useful even in the bounded setting.
The formula \( K_1 \) corresponding to the termination probability of \( P_1 \) does not have a unique fixed point, so the termination-based criteria for the soundness of Rule~\eqref{eq:wrong-principle} are not applicable, but Rule~\eqref{eq:proof-rule-unbounded} with well-chosen \( U \colon \mathbb{N} \to [0,1] \) actually works; see Example~\ref{ex:biased-random-walk-probability-using-u} that gives the exact lower bound for the termination probability of \( P_1 \) using Rule~\eqref{eq:proof-rule-unbounded}.
This stands in sharp contrast to most prior work, which could only provide approximate lower bounds for Example~\ref{ex:biased-random-walk-probability-using-u}.

If \( K \) does not have a unique fixed point, the above rule is not directly applicable.
This issue can be addressed by replacing  \( K \) with an under-approximation \( K' \le K \) that has a unique fixed point.
This simple idea allows us to apply the rule to a wider range of cases.
Moreover, several existing techniques for lower-bound verification, such as guard-strengthening~\cite{FengOOPSLA2023}, $\gamma$-scaling~\cite{UrabeLICS2017,TakisakaATVA2018}, and stochastic invariants~\cite{ChatterjeePOPL2017,ChatterjeeCAV2022}, can be seen as instances of this under-approximation technique.

Our method provides a unified rule applicable to a broad range of settings, as presented in Table~\ref{table:comparison-of-methods}.
Various methods have been proposed for providing lower bounds for different kinds of transformers \( K \), such as the weakest pre-expectation transformer~\cite{HarkPOPL2020}, the expected running-time transformer~\cite{HarkPOPL2020,WangPLDI2019}, and its higher-moment variants~\cite{WangPLDI2021a}.
While these methods are closely interrelated, their soundness proofs must be treated individually.
Furthermore, depending on the particular setting, distinct assumptions are required regarding program termination and supplementary conditions (such as \emph{conditional difference boundedness}~\cite{HarkPOPL2020} and \emph{bounded-update}~\cite{WangPLDI2021a}), necessitating subtle rule design and careful analysis.
Our rule is applicable to all the above-mentioned settings as well as to other transformers such as those soft/hard conditioning proposed by \citet{SzymczakSETSS2020} and \citet{OlmedoTOPLAS2018}.

Another interesting feature is that our rule is applicable to non-terminating programs.
The problem of providing lower bounds on termination probability is an important problem studied in \cite{ChatterjeeCAV2022,MajumdarPOPL2025}.
In this setting, program termination cannot, of course, be assumed.
Our method is capable of giving an exact lower bound on the termination probability of program \( P_1 \) (see Example~\ref{ex:biased-random-walk-probability-using-u}).
To the best of our knowledge, no other existing approach can provide an exact lower bound on the termination probability of the program \( P_1 \) within a single proof.\footnote{The completeness of the proof rules in \cite{ChatterjeeCAV2022,MajumdarPOPL2025} are proved, but their completeness results are of the following form: for a program \( P \) with the termination probability \( p \), for each \( n \), one can prove that \( P \) terminates with probability at least \( p - 1/n \).  Hence, to prove that \( P \) terminates with probability \( p \), these rules require a countable family of witnesses each providing a lower bound \( p - 1/n \) in general.}

\begin{table}[b]
    \caption{Methods for lower bound estimation of quantitative properties of probabilistic programs.
        ``CDB'' and ``BUD'' stand for Conditionally Difference Bounded and Bounded Update, respectively, which are conditions required for the soundness of some methods.}
    \label{table:comparison-of-methods}
    \vspace{-2.5ex}
    \small
    \begin{tabular}{r|ccccc|c|}
        & & & & \hphantom{} & \hphantom{${}^{(5)}$} & ours \\
        & \cite{HarkPOPL2020} & \cite{FengOOPSLA2023} & \cite{WangPLDI2021a,WangPLDI2019} & \cite{McIver2005} & \cite{TakisakaATVA2018} &
        \\ \hline \hline
        Weakest preexpectation & \checkmark & \checkmark & & \checkmark & \checkmark & \checkmark
        \\
        Expected runtime & \checkmark & \checkmark & \checkmark & & & \checkmark
        \\
        Higher moment of runtime & & & \checkmark & & & \checkmark
        \\
        Soft/Hard conditioning & & & & & & \checkmark
        \\ \hline
        Applicable to unbounded expectations & \checkmark & \checkmark & \checkmark & & & \checkmark
        \\
        Applicable to non-terminating programs & & (\checkmark)\rlap{${}^{*1}$} & & \checkmark & \checkmark & \checkmark
        \\
        CDB/BUD requirement free & & \checkmark & & \checkmark & \checkmark & \checkmark        \\ \hline
        Negative values & & & \checkmark & & & 
        \\ \hline
        Exact estimation of $P_1$ & & & & & & \checkmark
        \\
        Automation & \checkmark\rlap{${}^{*2}$} & & \checkmark & & \checkmark & \checkmark
    \end{tabular}
    \\[4pt]
        \begin{description}
            \item[*1] The target program does not need to satisfy AST/CDB, but one needs to transform the target program into one that satisfy AST.
            \item[*2] Implementation is not given in \citet{HarkPOPL2020}, but actually found in \cite{BatzTACAS2023}.
        \end{description}
\end{table}

\paragraph{Implementation}
We give a prototype implementation of the proposed method.
Our tool automatically infers a lower-bound of \( \mu K \) by finding an appropriate \( K' \le K \), proving \( \mu K' = \nu K' \) and giving \( X \) with \( X \le K'(X) \).
Actually, our tool does not solve the problem in this sequential manner.
Instead, it extracts the constraints that \( K' \) and \( X \) should satisfy and searches for a solution of the constraints by using a template-based approach.
We evaluated the tool on some examples from the literature~\cite{FengOOPSLA2023,HarkPOPL2020,OlmedoTOPLAS2018} and discussed its effectiveness.

\paragraph{Contributions.}
The contributions of this paper are summarized as follows.
\begin{itemize}
    \item We propose a new reasoning principle for estimating a lower bound of the least fixed point \( \mu K \).
        Its core is to prove ``termination'' in a somewhat non-standard sense, by using a generalized ranking supermartingale.
    \item We prove the correctness of our reasoning principle and demonstrate its usefulness by examples.
        We also compare our approach with existing ones.
    \item We also discuss a method to automate the lower-bound proof using the proposed principle.
        Our method is based on a template based approach.
    \item We give a prototype implementation of the proposed method and evaluate the effectiveness of our method via an experiment.
        The experiments confirm that our tool provides exact bounds for many examples in the literature, and it is a distinctive feature of our tool.
\end{itemize}

Finally, we briefly discuss the limitations of the proposed approach, dividing them into two categories: the limitations of the proof rule itself, and challenges related to implementation.

Regarding the proof rule proposed in this paper, the first limitation is its incompleteness.
This incompleteness stems from the reliance on ranking supermartingales; a counterexample can be constructed from a program that is almost-surely terminating but not positively almost-surely terminating.
Another limitation is that our method cannot directly establish \( \infty \) as a lower bound.
However, it is important to note that the incompleteness here only means that exact lower bounds cannot always be established in a single proof.
Our method is ``complete'' in the sense of \citet[Theorem~7]{AbateCAV2025} and \citet[Proof Rule~4.3 and Lemma~4.6]{MajumdarPOPL2025}, i.e., our method still provides lower bounds arbitrarily close to the exact value.
Our method does not support non-negative expectations. This limitation can be addressed to some extent by decomposing mixed-sign expectations \( X \) into positive and negative parts \( X = X^{+} - X^{-} \).

As for implementation, the current major limitation is the inability to handle exponential functions. As seen in many examples throughout the paper, exact lower-bound estimation often requires exponential expressions, and the current tool cannot provide exact bounds for such cases.

\paragraph{Outline.}
Section~\ref{sec:problem} gives a formal definition of the problem.
Section~\ref{sec:unique-fixed-points-unbounded} gives a sufficient condition for the uniqueness of fixed point.
Section~\ref{sec:reasoning-principle} introduces proof rules based on the results in Section~\ref{sec:unique-fixed-points-unbounded} and demonstrates their power by examples.
Section~\ref{sec:implementation} discusses the automation and reports an experimental result.
Section~\ref{sec:related} discusses related work, and we conclude in Section~\ref{sec:conc}.
Omitted details can be found in the
\iflong
appendix.
\else
the full version~\cite{arxiv}.
\fi

\section{Program Analysis via Fixed-Point Equations}
\label{sec:problem}
This section introduces a fixed-point calculus that serves as the target of lower-bound estimation in this paper.
Instead of directly reasoning about programs, we work with the fixed-point calculus because this approach provides a uniform framework for various analyses, such as termination probability, expected cost, and the second moment of cost.
Differences in the properties being analysed are considered only at the stage of translating programs into formulas.

After reviewing some basic facts about fixed points in Section~\ref{sec:pre:pre}, we introduce the fixed-point calculus in Section~\ref{sec:pre:calculus-definition}.
Section~\ref{sec:pre:language-definition} explains how our fixed-point calculus can be used to describe various properties of probabilistic programs.
A more general discussion of the expressiveness can be found in \citet{Kura2024}.
Readers who are already convinced that our fixed-point calculus is sufficiently expressive can safely skip Section~\ref{sec:pre:language-definition} (although some examples will be used later).

\begin{notation}\label{notation:extended-real}
	Let \( [0, \infty) := \{ r \in \mathbb{R} \mid r \ge 0 \} \) be the set of non-negative real numbers and \( [0,\infty] := [0,\infty) \cup \{ \infty \} \) be the set of \emph{extended non-negative real numbers}.
	The set \( [0,\infty] \) is a totally ordered set by setting \( x \le \infty \) for every \( x \in [0,\infty] \).
	The addition is extended to \( \infty \) by \( \infty + x = x + \infty = \infty \) for every \( x \in [0, \infty] \).
	The multiplication of \( r \in [0, \infty) \) and \( x \in [0, \infty] \) is given by
	\(
		0 \cdot \infty = 0
	\) and
	\(
		r \cdot \infty = \infty \ (r \neq 0)
	\)
	(the product of two extended real numbers will not be used).
	When \( x \ge y \), the subtraction \( x - y \) is defined as the minimum value \( z \in [0,\infty] \) such that \( x = y + z \), so \( \infty - \infty = 0 \) and \( \infty - z = \infty \) if \( z < \infty \).
\end{notation}

\subsection{Preliminaries on Fixed Points}\label{sec:pre:pre}
Let \( D \) be a set, such as the set of all states of a system.
We write \( \mathbb{E}(D) \) for the set \( D \to [0,\infty] \) of \emph{expectations} and \( \mathbb{E}_{\le 1}(D) \) for \( D \to [0,1] \).
An element of \( \mathbb{E}_{\le 1}(D) \) is called a \emph{\( 1 \)-bounded expectation}, and \( \mathbb{E}_{\le 1}(D) \) is regarded as a subset of \( \mathbb{E}(D) \). 
The sets \( \mathbb{E}(D) \) and \( \mathbb{E}_{\le 1}(D) \) are equipped with the point-wise order (i.e., \( f \le g \) if and only if \( f(d) \le g(d) \) for every \( d \in D \)).

A partially ordered set $(X, {\le})$ is an \emph{$\omega$-complete partial order} (or \emph{$\omega$cpo}) if it has a least upper bound $\sup_n x_n$ of any ascending chain $x_0 \le x_1 \le \cdots$ in $X$.
If an \( \omega \)cpo \( X \) has a minimum element, we write it as \( \bot \) and call it the \emph{bottom}.
All of \( [0,\infty] \), \( [0,1] \), \( \mathbb{E}(D) \) and \( \mathbb{E}_{\le 1}(D) \) are \( \omega \)cpos with bottom elements given by the constant zero function $\probzero$.
A function $K : X \to Y$ between $\omega$cpos is \emph{Scott continuous} if it is monotone and preserves the least upper bound of any ascending chain: $K(\sup_n x_n) = \sup_n K(x_n)$.

Given a function \( K : X \to X \) on a poset \( X \), the \emph{least fixed-point}
of \( K \) is the minimum
element in \( \{ x \in X \mid K(x) = x \} \).
If the least fixed-point of \( K \) exists, we write it as \( \mu K \).
A least fixed point does not exist in general, but the following sufficient condition is well known.
\begin{theorem}[Kleene's fixed-point theorem]
	If $(X, {\le})$ is an $\omega$cpo with a bottom element $\bot$ and $K : X \to X$ is a Scott continuous function, then $K$ has the least fixed point $\mu K = \sup_n K^n(\bot)$.
	Here \( K^n \) is defined by \( K^0(x) = x \) and \( K^{n+1}(x) = K(K^n(x)) \).
	\qed
\end{theorem}

By the duality, we obtain the following notions and results.
A partially ordered set $(X, {\le})$ is an \emph{$\omega$-cocomplete partial order} (or \emph{$\omega$ccpo}) if it has the greatest lower bound $\inf_n x_n$ of any descending chain $x_0 \ge x_1 \ge \cdots$ in $X$.
If an \( \omega \)ccpo has a maximum element, it is written as \( \top \) and called the \emph{top}.
A function $K : X \to Y$ between $\omega$ccpos is \emph{Scott cocontinuous} if it is monotone and preserves the greatest lower bound of any descending chain: $K(\inf_n x_n) = \inf_n K(x_n)$.
The maximum element in \( \{ x \in X \mid K(x) = x \} \) is called the \emph{greatest fixed-point} of \( K \) and written as \( \nu K \).
By the dual of Kleene's fixed-point theorem, if $(X, {\le})$ is an $\omega$ccpo with a top element $\top$ and $K : X \to X$ is a Scott cocontinuous function, then $K$ has a greatest fixed point given by $\nu K = \inf_n K^n(\top)$.

\begin{remark}
	The partially ordered set \(\mathbb{E}(D)\) has a stronger completeness property than $\omega$-completeness.
	That is, \(\mathbb{E}(D)\) is a complete lattice.
	This implies that any monotone function \( K : \mathbb{E}(D) \to \mathbb{E}(D) \) has both a least fixed point \( \mu K \) and a greatest fixed point \( \nu K \).
	Although Scott continuity and cocontinuity of $K$ are not required for the existence of these fixed points, we assume both in this paper, as we rely on the characterizations \(\mu f = \sup_n f^n(\bot)\) and \(\nu f = \inf_n f^n(\top)\) to provide a sufficient condition for the uniqueness of fixed points.
	This assumption does not restrict the applicability of our method as we will see in Section~\ref{sec:pre:calculus-definition} and Lemma~\ref{lem:continuity-of-interpretation}.
\end{remark}

A \emph{prefixed point} (resp.~\emph{postfixed point}) of $K : X \to X$ is an element $x \in X$ such that $K(x) \le x$ (resp.~$x \le K(x)$).
The Knaster--Tarski theorem relates pre- and post-fixed points with the least and greatest fixed-points.
\begin{theorem}[Knaster--Tarski theorem]
	Let \( K \colon \mathbb{E}(D) \to \mathbb{E}(D) \) be a monotone function.
	Then (a) if \( K(x) \le x \), then \( \mu K \le x \); and
	(b) if \( x \le K(x) \), then \( x \le \nu K \).
	\qed
\end{theorem}

\subsection{Fixed-Point Equation System}
\label{sec:pre:calculus-definition}
As mentioned at the beginning of this section, this paper focuses on $[0, \infty]$-valued logical formulas that involve least fixed-points.
While there are multiple approaches to formalizing such formulas, we adopt a representation based on systems of mutually-recursive equations.

Assume a collection \( \mathcal{D} \) of sets, of which an element is called a \emph{data type}.
Intuitively \( D \in \mathcal{D} \) is a data type in the target programming language, and typically \( \mathcal{D} = \{ \mathbb{N}, \mathbb{Z}, \mathbb{Q}, \mathbb{R}, \dots \} \).
However, our theory does not require any assumptions about \( \mathcal{D} \), so one can assume that an arbitrary complex data type belongs to \( \mathcal{D} \).
The only exceptions are the following cases: (1) in Section~\ref{sec:implementation} on implementation, we assume that data types in \( \mathcal{D} \) can be handled by the external solver on which our implementation relies (see Section~\ref{sec:input-format-for-implementation} for details); (2) to ease the presentation of the theoretical part, we shall assume that \( \mathcal{D} \) contains the Boolean type \( \mathbb{B} = \{ \mathbf{true}, \mathbf{false} \} \), the probability type \( [0,1] \) and the positive real number type \( [0,\infty) \), and \( \mathcal{D} \) is closed under finite cartesian product and disjoint union.
We use \( D \) as a metavariable ranging over data types (i.e., \( D \in \mathcal{D} \)), and \( x,y,z,\dots \) as variables ranging over elements of data types (i.e., \( x \in D \) for some \( D \in \mathcal{D} \)).

We also assume a collection \( \mathcal{OP} \) of operations on data types, which correspond to operations in the target programming language.
Each operation \( \mathit{op} \in \mathcal{OP} \) is associated with a type \( D_1 \times \dots \times D_n \to D \) (where \( D_1,\dots,D_n,D \in \mathcal{D} \)), meaning that it takes an \( n \)-tuple \( (x_1, \dots, x_n) \) of arguments of type \( x_i \colon D_i \) and returns an element \( \mathit{op}(x_1,\dots,x_n) \) of \(  D \).
Typically, the collection of operations has the summation \( + \) and multiplication \( \times \) on \( \mathbb{Z} \), as well as constants \( 0, 1 \in \mathbb{Z} \) (regarded as \( 0 \)-ary operations), but again our theory does not rely on any assumption on operations\footnote{Corresponding to the assumptions on \( \mathcal{D} \) in case (2), to simplify the presentation, we implicitly assume basic operations such as the projection \( D_1 \times D_2 \to D_i \) and injection \( D_i \to (D_1 + D_2) \) for every \( D_1,D_2 \in \mathcal{D} \) and \( i = 1, 2 \) belong to \( \mathcal{OP} \).}.

The set \( \mathit{Exp}_{D} \) of \emph{expressions} of type \( D \in \mathcal{D} \) is defined by the following grammar:
\begin{equation*}
	e \in \mathit{Exp}_D
	\qquad::=\qquad
	x
	\quad\mid\quad
	\mathit{op}(e_1,\dots,e_n),
\end{equation*}
where \( x \) is a variable of type \( D \), the operation \( \mathit{op} \) is of type \( D_1 \times \dots \times D_n \to D \), and \( e_i \in \mathit{Exp}_{D_i} \) for every \( i = 1,\dots,n \).
In particular, we refer to elements of \( \mathit{Exp}_{\mathbb{B}} \) as \emph{boolean expressions}, written \( \varphi \), and elements of \( \mathit{Exp}_{[0,\infty)} \) as \emph{non-negative real term}, written \( t \).
In a typical setting, \( \mathit{Exp}_{\mathbb{Z}} \) contains arbitrary arithmetic expressions.

\begin{definition}\label{def:fixed-point-equation}
	Let \( \mathcal{X} = \{ X_1, X_2, \dots, X_n \} \) be a finite set of \emph{quantitative predicate variables}.
	Each \( X_i \) is associated with its type \( D_{i,1} \times \dots \times D_{i,m_i} \to \Omega \) where \( D_{i,j} \in \mathcal{D} \) and \( \Omega \) is a special symbol, which intuitively means the type of quantitative propositions (i.e., \( \Omega = [0,\infty] \)).

	The set of \emph{quantitative formulas} over \( \mathcal{X} \) is defined by the following grammar:
	\begin{equation*}
		F \quad\coloneqq\quad X(\widetilde{e}) \mid t \mid F_1 + F_2 \mid t \cdot F \mid \ifexpr{\varphi}{F_1}{F_2},
	\end{equation*}
	where \( \widetilde{e} \) is a sequence of expressions (of expected types), \( t \in \mathit{Exp}_{[0,\infty)} \) is a non-negative real term, and \( \varphi \in \mathit{Exp}_{\mathbb{B}} \) is a boolean expression.
	A \emph{fixed-point equation system for expectations} (or just \emph{equation system}) is a set $E$ of equation in the following form:
	\begin{equation}
		E \quad=\quad \{ X_1(\widetilde{x_1}) =_{\mu} F_1,\quad X_2(\widetilde{x_2}) =_{\mu} F_2,\quad \dots, X_n(\widetilde{x_n}) =_{\mu} F_n \}
		\label{eq:hes}
	\end{equation}
	where $\widetilde{x_i}$ is a sequence of term variables (of expected types) and $F_i$ is a quantitative formula possibly containing variables \( \widetilde{x_i} \) and quantitative predicate variables \( X_1,\dots,X_n \) (but no other variables).
	Note that \( E \) contains an equation for each \( X_i \in \mathcal{X} \).
\end{definition}

The meaning of each construct should be clear (see the formal semantics given below).

Here we give some remarks on the syntax.
The addition $F_1 + F_2$ and the scalar multiplication $t \cdot F$ in quantitative formulas are often used to express the expectation of the form $\sum_i t_i \cdot F_i$ where $\sum_i t_i = 1$ (i.e.\ $\sum_i t_i \cdot F_i$ is the integral with respect to a discrete probability distribution), but the definition is more permissive since we can also use general weighted sums $\sum_i t_i \cdot F_i$ where $t_i \ge 0$ for each $i$ (i.e.\ integral with respect to an arbitrary discrete measure).
For example, we will make use of this generality when we apply our technique to the verification of higher moments of the cost (cf.~Example~\ref{ex:equation-systems-cost-moment}).
The expressive power of the fixed-point equation system for expectations shall be discussed in the next subsection.

\begin{remark}
	We briefly discuss features not supported by the logic in Definition~\ref{def:fixed-point-equation}. The most significant ones are: (1) continuous distributions, (2) disjunction and conjunction, and (3) mixed-sign expectations (i.e., expectations that can take negative values).

	As for continuous distributions, we do not support them for simplicity of the presentation.
	It is straightforward to extend our theoretical framework to support continuous distributions once we restrict a set of data types $\mathcal{D}$ to a set of \emph{measurable spaces} and the set $\mathbb{E}(D)$ of expectations to the set of \emph{measurable} expectations $D \to [0, \infty]$.
	A sketch of this extension can be found in \referappendix{Appendix}{F}{sec:integration}.

	On the contrary, supporting disjunction and conjunction has a technical challenge.
	The syntax of formulas is designed so that the semantics of formulas is affine in the sense of Definition~\ref{def:affineness} (Lemma~\ref{lem:affineness-of-interpretation}), and the affineness is heavily used in the current development.
	However, the disjunction and conjunction (interpreted as \( \max \) and \( \min \)) are not affine.
	We shall briefly discuss a way to handle \( \min \) in \referappendix{Appendix}{E.1}{sec:demonic-nondeterminism}.

	Finally, we do not support mixed-sign expectations due to inherent difficulties in defining least fixed point semantics in this setting~\cite{KaminskiLICS2017}.
	Nevertheless, mixed-sign expectations can be handled to some extent by decomposing a given expectation \( X : D \to \mathbb{R} \) into a pair of non-negative expectations \( X^+ : D \to [0,\infty] \) and \( X^- : D \to [0,\infty] \) such that \( X = X^+ - X^- \) (see \referappendix{Appendix}{G}{sec:experiment-negative-cost}).
\end{remark}

\paragraph{Semantics.}
Assume a fixed-point equation system \( E \) and suppose that \( X_i \) is of type \( D_i \to \Omega \) for simplicity.
Each quantitative formula $F_i$ defines a function $\interpret{F_i} : \prod_{j = 1}^n \mathbb{E}(D_j) \to \mathbb{E}(D_i)$, and the simultaneous least fixed point of the functions $\interpret{F_i}$ gives the least solution of the equation system $E$.
The formal definition is given as follows.
\begin{definition}\label{def:interpretation-fixed-point-equation}
	Assume the interpretations of non-negative real terms $\interpret{t} : D_i \to [0, \infty)$, boolean terms $\interpret{\varphi} : D_i \to \mathbb{B}$, and expressions $\interpret{e} : D_i \to D$ (for \( e \in \mathit{Exp}_D \)).
	The \emph{interpretation} of quantitative formulas $F_i$ is a function $\interpret{F_i} : \prod_{j = 1}^n \mathbb{E}(D_j) \to \mathbb{E}(D_i)$ inductively defined as follows.
	\begin{gather*}
		\interpret{X_j(e)}(\eta)(v) \coloneqq \eta_j(\interpret{e}(v))
		\qquad
		\interpret{t}(\eta)(v) \coloneqq \interpret{t}(v)
		\\
		\interpret{F + F'}(\eta)(v) \coloneqq \interpret{F}(\eta)(v) + \interpret{F'}(\eta)(v)
		\qquad
		\interpret{t \cdot F}(\eta)(v) \coloneqq \interpret{t}(v) \cdot \interpret{F}(\eta)(v)
		\\
		\interpret{\ifexpr{\varphi}{F}{F'}}(\eta)(v) \coloneqq\quad \begin{cases}
			\interpret{F}(\eta)(v) & \text{if $\interpret{\varphi}(v)$ is \( \mathbf{true} \)} \\
			\interpret{F'}(\eta)(v) & \text{if $\interpret{\varphi}(v)$ is \( \mathbf{false} \),}
		\end{cases}
	\end{gather*}
	where $\eta_j$ is the $j$-th component of $\eta \in \prod_{j = 1}^n \mathbb{E}(D_j)$.

	The \emph{interpretation} of the equation system~\eqref{eq:hes} is given as a function $\interpret{E} : \prod_{j = 1}^n \mathbb{E}(D_j) \to \prod_{j = 1}^n \mathbb{E}(D_j)$ defined as the combination of $\interpret{F_i}$, that means,
	\( \interpret{E}(\eta) \coloneqq (\interpret{F_1}(\eta), \dots, \interpret{F_n}(\eta)) \).
	The \emph{(least) solution} of the fixed-point equation system $E$ is the least fixed point of $\interpret{E} : \prod_{j = 1}^n \mathbb{E}(D_j) \to \prod_{j = 1}^n \mathbb{E}(D_j)$, which is denoted by $\mu \interpret{E}$.
\end{definition}

\begin{remark}\label{rem:simplify-expectation}
	Note that there exists a natural isomorphism $\prod_{j = 1}^n \mathbb{E}(D_j) \cong \mathbb{E}(\coprod_{j = 1}^n D_j)$, where \( \coprod_{j = 1}^n D_j \) is the disjoint union of \( D_1,\dots,D_n \).
	We write \( D \) for \( \coprod_{j = 1}^n D_j \) and often identify $\prod_{j = 1}^n \mathbb{E}(D_j)$ with $\mathbb{E}(D)$.
	This simplifies the notation: for example, $\interpret{E} : \mathbb{E}(D) \to \mathbb{E}(D)$.
\end{remark}

To justify the definition of \( \mu \interpret{E} \), we first show the existence of the least fixed-point of \( \interpret{E} \).
\begin{lemma}\label{lem:continuity-of-interpretation}
	The interpretation \( \interpret{F} \colon \mathbb{E}(D) \to \mathbb{E}(D') \) of a quantitative formula \( F \) is Scott continuous and Scott cocontinuous.
	The interpretation \( \interpret{E} \colon \mathbb{E}(D) \to \mathbb{E}(D) \) of an equation system \( E \) is Scott continuous and Scott cocontinuous.
\end{lemma}
\begin{proof}
	The latter immediately follows from the former.
	The former claim can be proved by induction on the structure of $F$.
	Note that \( \interpret{F} \colon \mathbb{E}(D) \to \mathbb{E}(D') \) is Scott (co)continuous if and only if \( \interpret{F}({-})(v') \colon \mathbb{E}(D) \to [0,\infty] \) is (co)continuous for every \( v' \in D' \) since \( \sup \) and \( \inf \) in \( \mathbb{E}(D') = (D' \to [0,\infty]) \) can be computed point-wise.
	So the claim follows from the fact that $({+}) : [0, \infty] \times [0, \infty] \to [0, \infty]$ and $r \cdot ({-}) : [0,\infty) \times [0, \infty] \to [0, \infty]$ (for every \( r \in [0,\infty) \)) are Scott continuous and Scott cocontinuous.
\end{proof}

\begin{corollary}
	\( \interpret{E} \colon \mathbb{E}(D) \to \mathbb{E}(D) \) has a least fixed-point.
	\qed
\end{corollary}

An important property of the interpretations of quantitative formulas
is the affineness.
\begin{definition}[Affineness]\label{def:affineness}
	A function \( K \colon \mathbb{E}(D) \to \mathbb{E}(D') \) is \emph{affine} if it is Scott continuous and cocontinuous and moreover satisfies
	\begin{equation*}
		K(\alpha \eta_1 + (1-\alpha) \eta_2)
		\quad=\quad
		\alpha K(\eta_1) + (1-\alpha) K(\eta_2)
	\end{equation*}
	for every \( \eta_1, \eta_2 \in \mathbb{E}(D) \) and \( \alpha \in [0,1] \).
\end{definition}
\begin{lemma}\label{lem:affineness-of-interpretation}
	The interpretations \( \interpret{F} \) and \( \interpret{E} \) are affine.
\end{lemma}
\begin{proof}
	The affineness of \( \interpret{E} \) follows from the affineness of \( \interpret{F} \).
	We prove the affineness of \( \interpret{F} \colon \mathbb{E}(D) \to \mathbb{E}(D') \).
	Since the sum in \( \mathbb{E}(D') \) is defined as the point-wise sum, it suffices to prove that \( \interpret{F}({-})(v') \) is affine for every \( v' \in D' \).
	This claim can be easily proved by induction on \( F \).
\end{proof}

\begin{remark}
	It is important to distinguish between the affineness of the interpretations of quantitative formulas and the affineness of operations on data types.
	Lemma~\ref{lem:affineness-of-interpretation} shows that the interpretations of quantitative formulas are always affine, but it does not assume anything about the ``affineness'' of operations involved in the formulas.
	So Lemma~\ref{lem:affineness-of-interpretation} holds even if \( \mathcal{OP} \) has a ``non-affine'' operation such as multiplication \( ({-}) \times ({-}) \colon \mathbb{Z} \times \mathbb{Z} \to \mathbb{Z} \) and power \( ({-})^{({-})} \colon \mathbb{Z} \times \mathbb{N} \to \mathbb{Z} \).

	For example, \( F = x^2 \cdot X(e^x) \) is affine in the sense of Definition~\ref{def:affineness} because it is affine with respect to $X$.
	On the other hand, \( F = 2 \cdot X(x - 1) \cdot X(x + 1) \) is not affine and actually not a quantitative formula in the sense of Definition~\ref{def:fixed-point-equation}, even though terms (e.g.\ $x - 1$) in $F$ are affine.

	For readers familiar with semantics, this point can be explained as follows.
	The affineness of the interpretations of formulas in Lemma~\ref{lem:affineness-of-interpretation} corresponds, in terms of programming languages, to the affineness of control flow, i.e.~the property that a continuation is used at most once.
	This property holds in a wide range of programming languages.
	Moreover, since it concerns the usage pattern of continuations, it is independent of operations on data types.
\end{remark}

\subsection{A Probabilistic Programming Language and its Verification via Equation System}
\label{sec:pre:language-definition}
This subsection demonstrates the expressive power of the fixed-point equation systems introduced in Section~\ref{sec:pre:calculus-definition}.
To this end, we introduce a probabilistic programming language and show that various properties of programs
can be described in terms of fixed-point equation systems.
Readers who are already convinced that our fixed-point calculus is sufficiently expressive can safely skip this subsection (although some examples will be used later).

We consider probabilistic programs in the \emph{probabilistic guarded command language} (\emph{pGCL}).
The syntax is defined as follows.
\[ c \quad\coloneqq\quad \mathsf{skip} \mid c_1; c_2 \mid \assign{x}{e} \mid \pbranch{c_1}{p}{c_2} \mid \ifstmt{\varphi}{c_1}{c_2} \mid \while{\varphi}{c} \]
Here, $e \in \mathit{Exp}_D$ is an expression, $\varphi \in \mathit{Exp}_{\mathbb{B}}$ is a boolean expression, and $p \in \mathit{Exp}_{[0, 1]}$ is an expression representing a probability.
The meaning of each statement is standard (see e.g.~\cite{FengOOPSLA2023}).
For example, $\pbranch{c_1}{p}{c_2}$ is the probabilistic branching that executes $c_1$ with probability $p$ and $c_2$ with probability $1 - p$.
We see that various properties and quantities on a program \( c \) can be expressed as solutions of equation systems.

\subsubsection{Weakest Precondition Transformer}
The \emph{weakest preexpectation transformer}~\cite{McIver2005} is studied as a probabilistic counterpart of the weakest precondition transformer.
Assume a program \( c \) using variables \( x_1,\dots,x_n \) with type \( x_i \colon D_i \) and let \( D = \prod_{i = 1}^n D_i \).
The weakest preexpectation transformer $\mathrm{wp}[c] : \mathbb{E}(D) \to \mathbb{E}(D)$ of the program $c$ takes a postexpectation $f \in \mathbb{E}(D)$ and returns the weakest preexpectation, which is the expected value of $f$ after executing $c$.
The weakest preexpectation for a while loop $\mathrm{wp}[\while{\varphi}{c}](f) = \mu \Phi$ is the least fixed point of $\Phi(X) \coloneqq \ifexpr{\interpret{\varphi}}{\mathrm{wp}[c](X)}{f}$.
The concrete definition of $\mathrm{wp}[c]$
is found in \referappendix{Appendix}{A}{appx:wp-formally}.

It is not difficult to see that the weakest preexpectation is definable as the solution of a fixed-point equation system.
In the definition of $\mathrm{wp}[c]$, pre- and post-expectations are expressed as nested least fixed-points, for instance $F_0(\mu X_1. F_1(X_1, \mu X_2. F_2(X_2, \dots)))$.
Such nested least fixed-points can be equivalently represented as a pair $(F, E)$ where $F$ is a quantitative formula expressed in terms of the solution of the least fixed-point equations in $E$, for example $F = F_0(X_1)$ and $E = \{ X_1 =_{\mu} F_1(X_1, X_2), X_2 =_{\mu} F_2(X_2, \dots), \dots \}$ corresponding to the above example.
Based on this observation, we reformulate the definition of $\mathrm{wp}[c]$ in terms of fixed-point equation systems.
Specifically, we define a translation $\mathrm{wp}'[c]$ that takes and returns a pair of a quantitative formula $F$ and a set of least fixed-point equations $E$.
For example, we have $\mathrm{wp}'[\while{\varphi}{c}](F, E) = (X(\widetilde{x}) , E' \cup \{ X(\widetilde{x}) =_{\mu} \ifexpr{\varphi}{F'}{F} \})$ when \( (F', E') = \mathrm{wp}'[c](X(\widetilde{x}), E) \), which corresponds to  $\mathrm{wp}[\while{\varphi}{c}](f) = \mu \Phi$ described above.
See \referappendix{Appendix}{A}{appx:wp-formally} for the complete list of rules.
For each non-negative real term (i.e.\ postexpectation) $t$, we have $\mathrm{wp}[c](t) = \interpret{F}(\mu \interpret{E})$ where $(F, E) = \mathrm{wp}'[c](t, \emptyset)$.
That is, the weakest preexpectation $\mathrm{wp}[c](t)$ is the \( W \) component of the solution of \( E \cup \{ W(\tilde{x}) =_{\mu} F \} \) (where \( W \) is a fresh quantitative predicate variable and \( \tilde{x} = (x_1,\dots,x_n) \) is the list of variables used in \( c \)).

\begin{example}\label{ex:biased-random-walk}
	Consider the following program of a biased random walk.
	\[ c_{\mathrm{rw}} \qquad\equiv\qquad \while{x > 0}{\pbranch{\assign{x}{x - 1}}{1/3}{\assign{x}{x + 1}}} \]
	For any postexpectation $f : \mathbb{Z} \to [0, 1]$, the weakest preexpectation $\mathrm{wp}[c_{\mathrm{rw}}](f)$ is given as the least solution $X : \mathbb{Z} \to [0, 1]$ for the following fixed-point equation:
	\begin{equation}
		X(x) \quad=_{\mu}\quad \ifexpr{x > 0}{\tfrac{1}{3} X(x - 1) + \tfrac{2}{3} X(x + 1)}{f(x)}
		\label{eq:biased-random-walk}
	\end{equation}
	When $f = \probone$ is the constant function, then $\mathrm{wp}[c_{\mathrm{rw}}](f)$ gives the termination probability.
	\qed
\end{example}

\subsubsection{Expected Cost Analysis}
The expected running time or other expected costs of a program are also properties of interest.
Here, we focus on expected running time.
To define the notion of running time, we add a new command \( \tickstmt \) to the programming language:
$c \coloneqq \cdots \mid \tickstmt$.
This command increments the running time by one (and other commands do not change the time).

The expected running time is characterized by the \emph{expected runtime transformer} $\mathrm{ert}[c]$~\cite{KaminskiJACM2018}.
In the current setting, \( \mathrm{ert} \) is similar to \( \mathrm{wp} \) but \( \mathrm{ert}[\tickstmt](f) := f + 1 \).

The expected runtime transformer can be straightforwardly expressed as the solution of a fixed-point equation system by the same reason as \( \mathrm{wp} \).
The translation $\mathrm{ert}'[c]$ is defined almost in the same way as $\mathrm{wp}'[c]$ except that $\mathrm{ert}'[c]$ is extended for the tick operator:
$\mathrm{ert}'[\tickstmt](F, E) \coloneqq (F + 1, E)$.

\begin{example}\label{ex:biased-random-walk-ert-definition}
	Consider the lower bound of the expected runtime of the biased random walk.
	\begin{equation}
		c_{\mathrm{rw}'} \quad=\quad \while{x > 0}{\tickstmt;\ (\pbranch{\assign{x}{x - 1}}{2/3}{\assign{x}{x + 1}})}
		\label{eq:biased-random-walk-ert-program}
	\end{equation}
	The expected running time is given as the least solution for the following equation.
	\begin{equation}
		X(x : \mathbf{int}) \quad=_{\mu}\quad \ifexpr{x > 0}{\tfrac{2}{3} X(x - 1) + \tfrac{1}{3} X(x + 1) + 1}{0}
		\label{eq:biased-random-walk-ert}
	\end{equation}
	We write \( E'_{\mathrm{rw}} \) for the equation system consisting only of the above equation.
	\qed
\end{example}

\subsubsection{Cost Moment Analysis}
Higher moments of the running time (i.e.\ the expected value of $k$-th power of the running time) of a probabilistic program are studied in~\cite{KuraTACAS2019,AguirreMSCS2022} as an extension of the expected runtime transformer.

We illustrate the logical representation of the second moment of the runtime.
The translation $\mathrm{rt}^{(2)}[c]$ takes and returns a tuple of \emph{two} quantitative formulas $F_1, F_2$ and a set of equations $E$.
Intuitively, $F_1$ is the first moment of the runtime and $F_2$ is the second moment.
In most cases, the translation is defined by applying $\mathrm{ert}'[c]$ component-wise.
For example, the translation for $\assign{x}{e}$ is defined as follows:
\[ \mathrm{rt}^{(2)}[\assign{x}{e}]((F_1, F_2), E) \quad\coloneqq\quad ((F_1[e/x], F_2[e/x]), E) \]
The only exception for the component-wise definition is the case for the $\tickstmt$ command, which is defined by the binomial expansion.
\[ \mathrm{rt}^{(2)}[\tickstmt]((F_1, F_2), E) \quad\coloneqq\quad ((F_1 + 1, F_2 + 2 F_1 + 1), E) \]

\begin{example}\label{ex:equation-systems-cost-moment}
	Consider the second moment of the runtime of the biased random walk~\eqref{eq:biased-random-walk-ert-program}.
	We obtain the following equation system by applying the above translation:
	\begin{align}
		X_1(x) \quad&=_{\mu}\quad \ifexpr{x > 0}{\tfrac{2}{3} X_1(x - 1) + \tfrac{1}{3} X_1(x + 1) + 1}{0} \\
		X_2(x) \quad&=_{\mu}\quad \ifexpr{x > 0}{\tfrac{2}{3} X_2(x - 1) + \tfrac{1}{3} X_2(x + 1) + 2 \left( \tfrac{2}{3} X_1(x - 1) + \tfrac{1}{3} X_1(x + 1) \right) + 1}{0}
	\end{align}
	The \( X_1 \) and \( X_2 \) components of the solution give the first and the second moment of the runtime, respectively.
	Note that the right-hand-side of \( X_2 \) contains a non-probabilistic weighted sum (i.e., the sum \( \tfrac{2}{3} + \tfrac{1}{3} + \tfrac{4}{3} + \tfrac{2}{3} \) of coefficients for quantitative predicate variables \( X_1 \) and \( X_2 \) exceeds \( 1 \)).
	\qed
\end{example}

\subsubsection{Soft/Hard Conditioning}\label{sec:soft-hard-conditioning}
Extensions of the weakest preexpectation transformer for soft/hard conditioning are also studied~\cite{SzymczakSETSS2020,OlmedoTOPLAS2018}.
We consider only soft conditioning for simplicity and provide a translation to equation systems.
Here, pGCL is extended with $\scoring{e}$, which models soft conditioning.
The statement $\scoring{e}$ scales the probability of the current execution trace by multiplying $e$ where $e$ is $[0, 1]$-valued expression.
Hard conditioning is a special case of soft conditioning where $e$ is given by the Iverson bracket $[\varphi]$ for a boolean expression $\varphi$.
The conditional weakest preexpectation is defined as $\mathrm{cwp}[c](f) = \mathrm{cwp}_1[c](f) / (\probone - \mathrm{cwp}_2[c](\probzero))$ where $\mathrm{cwp}_1[c]$ and $\mathrm{cwp}_2[c]$ are defined as least fixed points.
The corresponding translation $\mathrm{cwp}'_1[c](F, E)$ and $\mathrm{cwp}'_2[c](F, E)$ to equation systems is defined similarly to the translation of the weakest preexpectation with the following modification:
\[ \mathrm{cwp}'_1[\scoring{e}](F, E) \coloneqq (e \cdot F_1, E_1) \qquad \mathrm{cwp}'_2[\scoring{e}](F, E) \coloneqq ((1 - e) + e \cdot F_2, E_2) \]
where $(F_i, E_i) = \mathrm{cwp}'_i[c](F, E)$ for $i = 1, 2$.
As a side note, $\mathrm{cwp}_2$ is originally defined as the greatest fixed point, but here we define it as the least fixed point by applying $1 - ({-})$.

\section{Uniqueness of Fixed Points}
\label{sec:unique-fixed-points-unbounded}
As suggested in Introduction, our approach provides a lower bound \( \eta \) for \( \mu \interpret{E} \) based on two components:
a ranking argument establishing \( \mu \interpret{E} = \nu \interpret{E} \), and an invariant \( \eta \) of \( E \) (i.e., \( \eta \le \interpret{E}(\eta) \)).
This section formally describes the ranking argument and proves its correctness.

\subsection{Restricting the Domain and the Codomain}\label{sec:restriction-by-prefixed-point}
As we saw in Introduction, the least and greatest fixed-points do not necessarily coincide even if the underlying system is almost-surely terminating.
Consider, for example, \( E'_{\mathrm{rw}} \) in Example~\ref{ex:biased-random-walk-ert-definition} and let \( K = \interpret{E'_{\mathrm{rw}}} \).
The fixed-point equation \( E'_{\mathrm{rw}} \) corresponds to the expected runtime of a random work biased toward \( 0 \), so the underlying program is almost-surely terminating.
However, it has at least two fixed points: $X_1(x) = \ifexpr{x > 0}{3 x}{0}$ and $X_2(x) = \ifexpr{x > 0}{3x + (2^x-1)}{0}$.
Its greatest fixed point is \( \inf_n K^n(\infty) \), which is \( X_3(x) = \ifexpr{x > 0}{\infty}{0} \).

Our idea for overcoming this issue is to replace \( \mathbf{\infty} \), the maximum element in \( \mathbb{E}(D) \), with another well-behaved element \( u \in \mathbb{E}(D) \) and consider \( \inf_n K^n(u) \) instead of \( \inf_n K^n(\mathbf{\infty}) \).
Of course, not just any choice of \( u \) will work.
In particular, since the correctness proof of our reasoning principle relies on Knaster-Tarski Theorem for the greatest fixed point, \( \inf_n K^n(u) \) must, in some sense, behave like the iterative computation of a greatest fixed point.

Fortunately, the condition for \( \inf_n K^n(u) \) to be seen as the computation of a greatest fixed point turns out to be remarkably simple: it is just \( K(u) \le u \).
For \( u \in \mathbb{E}(D) \), let \( \mathbb{E}_{\le u}(D) \) be the sub-poset given by \( \{ \eta \in \mathbb{E}(D) \mid \eta \le u \} \).
\begin{lemma}\label{lem:restrict-domain-codomain}
	Let \( K : \mathbb{E}(D) \to \mathbb{E}(D) \) be a Scott continuous and cocontinuous function.
	If \( u \) is a prefixed point of \( K \), i.e., \( K(u) \le u \), then the restriction \( \restrictby{K}{u} \) of \( K \) to \( \mathbb{E}_{\le u}(D) \) is well-defined and \( \inf_n K^n(u) = \nu \restrictby{K}{u} \).
\end{lemma}
\begin{proof}
	By monotonicity, if $\eta \le u$, then $K(\eta) \le K(u) \le u$.
	So \( \restrictby{K}{u} \) is well-defined.
	We have \( \inf_n K^n(u) = \nu \restrictby{K}{u} \) since \( u \) is the maximum element of the restricted domain \( \mathbb{E}_{\le u}(D) \).
\end{proof}

\subsection{Ranking Supermartingales for Unique Fixed Points}\label{sec:generalised-ranking-supermartingales}
The remaining challenge is how to show that \( \inf_n K^n(u) = \sup_n K^n(\probzero) \) where \( u \) is a prefixed point of \( K \) and \( \probzero \) is the zero constant function.
Given an affine function \( K \), we write its linear part as \( \Diff K \colon \mathbb{E}(D) \to \mathbb{E}(D) \), which is defined by
$(\Diff K)(\eta) \coloneqq K(\eta) - K(\probzero)$.
Note that if \( u \) is a prefixed point of \( K \), then \( u \) is also a prefixed point of \( \Diff K \) since \( \Diff K(u) \le K(u) \le u \).
\begin{lemma}\label{lem:affine-k-induces-linear-knt}
	If\/ \( K \) is affine,	then \( \nonterm{K} \) is linear, i.e., \( \nonterm{K}(\eta_1 + \eta_2) = \nonterm{K}(\eta_1) + \nonterm{K}(\eta_2) \) and \( \nonterm{K}(\alpha \eta) = \alpha \nonterm{K}(\eta) \) for every \( \eta,\eta_1,\eta_2 \in \mathbb{E}(D) \) and \( \alpha \in [0,\infty) \).
	\qed
\end{lemma}
A sufficient condition for \( \inf_n K^n(u) = \sup_n K^n(\probzero) \) is the existence of \emph{\( u \)-ranking supermartingale}.
\begin{definition}
	We say $r : D \to [0, \infty)$ is a \emph{$u$-ranking supermartingale}\footnote{This is an abuse of terminology, as $r$ is not necessarily related to a stochastic process in general.}
	w.r.t.~$\nonterm{K}$ if it satisfies
	\begin{equation}
		\nonterm{K}(r) + u \le r.
		\label{eq:generalised-supermartingale}
	\end{equation}
\end{definition}

\begin{remark}
	We explain the connection between a ranking supermartingale for program termination~\cite{ChakarovCAV2013} and \( u \)-ranking supermartingale in the above definition.
	Consider the program \( c_{\mathrm{rw}} \) in Example~\ref{ex:biased-random-walk-ert-definition}.
	A \emph{ranking supermartingale} for \( c_{\mathrm{rw'}} \) is a function \( r \colon \mathbb{N} \to [0, \infty) \) such that \( \tfrac{2}{3} r(x-1) + \tfrac{1}{3} r(x+1) \le r(x) - 1 \) for every \( n > 0 \).
	Let \( K'_{\mathrm{rw}} = \interpret{E'_{\mathrm{rw}}} \).
	Then
	\begin{equation*}
		(\Diff K'_{\mathrm{rw}})(X)(x) \quad=\quad \ifexpr{x > 0}{\tfrac{2}{3} X(x - 1) + \tfrac{1}{3} X(x + 1)}{0},
	\end{equation*}
	so \( (\Diff K'_{\mathrm{rw}})(r) + \probone \le r \) if and only if \( r \) is a ranking supermartingale for \( c_{\mathrm{rw'}} \).
	This observation can be generalized: when \( K \) is the expected runtime transformer or weakest preexpectation transformer of a program, a \( \probone \)-ranking supermartingale for \( K \) coincides with a ranking supermartingale for the program.
	The above definition is a generalization applicable to any kind of transformer \( K \) and an upper-bound \( u \) of the least fixed point.
\end{remark}

It was a pleasant surprise to us that this simple attempt actually works.

\begin{proposition}\label{prop:diff-gfp-lfp-for-affine-k}
	Let \( K \colon \mathbb{E}(D) \to \mathbb{E}(D) \) be an affine function on \( \mathbb{E}(D) \).
	Then, for every \( n \ge 1 \), we have \( K^{n}(u) - K^n(\probzero) = (\nonterm{K})^n(u) \).
\end{proposition}
\begin{proof}
	By induction on $n$.
	The case \( n = 1 \) follows from the definition.
	If $K^n(\eta) = (\nonterm{K})^n(\eta) + K^n(\probzero)$, by the linearity of $\nonterm{K}$,
	\begin{align}
		K^{n + 1}(\eta) &= (\nonterm{K})(K^n(\eta)) + K(\probzero) \\
		&= \nonterm{K}((\nonterm{K})^n(\eta)) + (\nonterm{K})(K^n(\probzero)) + K(\probzero) \\
		&= \nonterm{K}^{n+1}(\eta) + K^{n + 1}(\probzero)
		\qedhere
	\end{align}
\end{proof}

\begin{theorem}\label{thm:unique-fixed-point-by-ranking-supermartingale}
	Let \( K : \mathbb{E}(D) \to \mathbb{E}(D) \) be an affine function and \( u : D \to [0, \infty) \).
	If \( \nonterm{K} \) has a \( u \)-ranking supermartingale, then \( \lim_n (\nonterm{K})^n(u) = \probzero \) (pointwise).
\end{theorem}
\begin{proof}
	By contradiction.
	Assume $\lim_{n \to \infty} (\nonterm{K})^n(u)(d) \neq 0$\footnote{Here, we extend the standard topology of \( [0, \infty) \) to \( [0, \infty] \) by treating \( \infty \) as a point at infinity (i.e. a neighbourhood base for \( \infty \) is given by intervals of the form \( [a, \infty] \) for $a \in [0, \infty)$).
	This ensures the well-definedness of the limit \( \lim_n (\nonterm{K})^n(u)(d) \) even when the sequence \( \{ (\nonterm{K})^n(u)(d) \}_n \) contains \( \infty \).
	Alternatively, one may argue by a case analysis on whether \( \{ (\nonterm{K})^n(u)(d) \}_n \) contains \( \infty \) and show that in either case, the sum \( \sum_{n = 0}^{\infty} (\nonterm{K})^n(u)(d) \) is infinite.} for some $d \in D$.
	Then, there exists $\epsilon > 0$ such that for infinitely many $n$, we have $(\nonterm{K})^n(u)(d) > \epsilon$.
	So \( \sum_{n = 0}^{\infty} (\nonterm{K})^n(u)(d) = \infty \).
	By monotonicity and linearity (Lemma~\ref{lem:affine-k-induces-linear-knt}) of \( \nonterm{K} \), we have the following inequality.
	\begin{equation*}
		\textstyle
		r\ \ge\ \nonterm{K}(r) + u\ \ge\ (\nonterm{K})^2(r) + (\nonterm{K})(u) + u\ \ge \dots \ge\ (\nonterm{K})^n(r) + \sum_{i = 0}^{n - 1} (\nonterm{K})^i(u)
	\end{equation*}
	Taking the limit as $n \to \infty$, we obtain $r \ge \sum_{n = 0}^{\infty} \nonterm{K}^n(u)$.
	However, this leads to a contradiction:
	$\infty > r(x) \ge \sum_{n = 0}^{\infty} (\nonterm{K})^n(u)(x) = \infty$.
\end{proof}

\begin{corollary}\label{cor:unique-fixed-point-by-ranking-supermartingale}
	Let $u : D \to [0, \infty)$ be a prefixed point of an affine function $K$.
	If\/ \( \nonterm{K} \) has a \( u \)-ranking supermartingale,
	the restriction \( \restrictby{K}{u} \)
	has a unique fixed point.
\end{corollary}
\begin{proof}
	Note that the greatest fixed-point of \( \restrictby{K}{u} \) is \( \inf_n K^n(u) \).
	By Proposition~\ref{prop:diff-gfp-lfp-for-affine-k}, we have \( K^n(u) = K^n(\probzero) + (\Diff K)^n(u) \),
	so \( \nu \restrictby{K}{u} = \mu \restrictby{K}{u} + \lim_n (\Diff K)^n(u) \)\footnote{Since $K^n(u)$ and $K^n(\probzero)$ are finite for every $n$, \( \inf_n K^n(u) \) and \( \sup_n K^n(\probzero) \) are equal to the topological limits \( \lim_n K^n(u) \) and \( \lim_n K^n(\probzero) \), respectively.} by taking the limit of \( n \to \infty \).
	We have \( \lim_n (\nonterm{K})^n(u) = 0 \) by Theorem~\ref{thm:unique-fixed-point-by-ranking-supermartingale}.
	Hence \( \nu \restrictby{K}{u} = \mu \restrictby{K}{u} \).
\end{proof}

\section{Reasoning Principle and Examples}\label{sec:reasoning-principle}
In this section, we propose proof rules for establishing lower bounds of least fixed points and examine its proving power.
The results from the previous sections give sufficient conditions for the uniqueness of fixed-points, which, together with the Knaster--Tarski theorem, yields rules for lower-bounding least fixed-points.
However, in general, the fixed points of a given equation system \( E \) may not be unique, and it may not be easy to find a suitable \( u \) such that \( E \) has a unique fixed-point in the restricted domain \( \mathbb{E}_{\le u}(D) \).
To address this difficulty, we first consider transforming the equation system \( E \) into a more tractable equation \( E' \) that under-approximates \( E \), and then apply the unique-fixed-point-based lower-bound proof rules to \( E' \).
This idea of under-approximating \( E \) can also be found in the literature on lower-bound verification, for example in guard-strengthening~\cite{FengOOPSLA2023}, $\gamma$-scaling~\cite{UrabeLICS2017,TakisakaATVA2018}, and stochastic invariants~\cite{ChatterjeePOPL2017,ChatterjeeCAV2022}.
Incorporating under-approximation of \( E \) also provides a clearer basis for comparison with these approaches, as discussed in more detail in \referappendix{Appendix}{D}{sec:simulating-existing-techniques-by-underapproximation}.

\subsection{Proof Rule}
\begin{theorem}\label{thm:unbounded-reasoning-principle}
	Let $E$ be a fixed-point equation system.
	We have \( \eta' \le \mu \interpret{E} \) if there exist:
	\begin{enumerate}
		\item[$(a)$] a fixed-point equation system \( E' \) such that \( \interpret{E'} \le \interpret{E} \),
		\item[$(b)$] a prefixed point \( u \) of\/ \( \interpret{E'} \), i.e., \( \interpret{E'}(u) \le u \),
		\item[$(c)$] a \( u \)-ranking supermartingale \( r : D \to [0, \infty) \) of\/ \( \nonterm{\interpret{E'}} \), and
		\item[$(d)$] an invariant \( \eta' \in \mathbb{E'}_{\le u}(D) \) of\/ \( \interpret{E'} \), i.e., \( \eta' \le \interpret{E'}(\eta') \) with \( \eta' \le u \).
	\end{enumerate}
\end{theorem}
\begin{proof}
	Using Corollary~\ref{cor:unique-fixed-point-by-ranking-supermartingale} with $(b)$,
	\begin{equation}
		(d)
		\:\xLongrightarrow{\textit{Knaster-Tarski}}\:
		\eta' \le \nu \restrictby{\interpret{E'}}{u}
		\:\xLongrightarrow{(c)}\:
		\eta' \le \mu \restrictby{\interpret{E'}}{u}
		\:\Longrightarrow\:
		\eta' \le \mu \interpret{E'}
		\:\xLongrightarrow{(a)}\:
		\eta' \le \mu \interpret{E}.
	\end{equation}
	Here, we use \( \mu \restrictby{\interpret{E'}}{u} = \lim_n \restrictby{\interpret{E'}}{u}^n(\probzero) = \lim_n \interpret{E'}^n(\probzero) = \mu \interpret{E'} \).
\end{proof}

\begin{remark}\label{rem:on-the-motivating-classical-rule}
	The rule in Theorem~\ref{thm:unbounded-reasoning-principle} corresponds to a classical termination proof method for non-probabilistic programs, which originally motivated this study.
	A naive adaptation of this rule to the classical setting, where fixed-point expressions take values in \(\{0,1\}\) rather than \([0,1]\) or \([0,\infty]\), proceeds as follows: Given a fixed-point expression \( E \), (a) provide an under-approximation \( \interpret{E'} \le \interpret{E} \), (b) supply a ranking function \( r \colon D \to \mathbb{N} \) that prove the ``termination'' of \( E' \), and (c) find a pre-fixed point \( \eta' \le \interpret{E'}(\eta') \) of \( E' \).
	However, in the classical case, this method contains a certain redundancy.
	Once a ``ranking function candidate'' \( r \colon D \to \mathbb{N} \) is fixed, there exists the best under-approximation \( E^{(r)} \) of \( E \) for which \( r \) is a ranking function.
	Then the above procedure can be reformulated as: (b') Fix a ranking function candidate \( r \colon D \to \mathbb{N} \), and (c') find a pre-fixed point \( \eta' \le \interpret{E^{(r)}}(\eta') \) of \( E^{(r)} \).
	This refined rule coincides with the well-known rule in classical termination analysis~\cite{Cook2005,Cook2006} (or its reformulation~\cite{Unno2023} in fixed-point logic).
	Our rule in Theorem~\ref{thm:unbounded-reasoning-principle} was derived by essentially reversing the direction of this argument.
\end{remark}

\subsection{Examples}\label{sec:underapproximating-monotone-functions}\label{sec:1ufp-reasoning}
We begin by examining the case where \( u = \probone \).
In our setting, finding a ``good'' under-approximation $E'$ is a non-trivial task.
We explain two approaches by examples.

\begin{example}[guard-strengthening]
	Consider the termination probability of the program in Example~\ref{ex:biased-random-walk} (i.e.\ we set the postexpectation to $f = \probone$), and let $K$ be the monotone function for~\eqref{eq:biased-random-walk}.
	The guard-strengthening~\cite{FengOOPSLA2023} provides a way to under-approximate the function $K$
	in Example~\ref{ex:biased-random-walk}.
	We replace the then-branch with $0$ if $x \ge M$ where \( M \) is some constant.
	\[ K'(X)(x) \quad\coloneqq\quad \ifexpr{0 < x}{(\ifexpr{x < M}{\tfrac{1}{3} X(x - 1) + \tfrac{2}{3} X(x + 1)}{0})}{1} \]
	It is easy to give a \( \probone \)-ranking supermartingale for \( \nonterm{K'} \)
	and thus, $K'$ has a unique fixed point in \(\mathbb{E}_{\le 1}(D)\), and any lower bound of $\nu K' = \mu K'$ gives a lower bound of $\mu K$.
	\qed
\end{example}

\begin{example}[subtracting probabilities]\label{ex:biased-random-walk-epsilon-subtraction}
	Another way to under-approximate the function $K$ in Example~\ref{ex:biased-random-walk} with $f = \probone$ is to subtract probabilities.
	Let $K'_{\epsilon} \le K$ be an approximation given by
	\[ K'_\epsilon(X)(x) \quad\coloneqq\quad \ifexpr{x > 0}{\tfrac{1}{3} X(x - 1) + \left(\tfrac{2}{3} - \epsilon \right) X(x + 1)}{1} \qquad \text{where $0 < \epsilon \le 2/3$.} \]
	This function corresponds to the random walk in Example~\ref{ex:biased-random-walk} with a small probability $\epsilon$ of abortion for each iteration.
	The modified random walk is almost surely terminating and \( K'_\epsilon \) has a \( \probone \)-ranking supermartingale,
	so $K'_\epsilon$ has a unique fixed point in \( \mathbb{E}_{\le 1}(\mathbb{N}) \).
	An invariant \( \eta \) of \( K'_{\epsilon} \) is given by
	\[ \eta(x) = \ifexpr{x > 0}{a^x}{1} \qquad \text{where } a = (3 - \sqrt{1 + 12 \epsilon}) / (4 - 6 \epsilon) \]
	If we take the limit $\epsilon \to 0$, then we obtain $a = 1/2$, which gives the exact bound.
	\qed
\end{example}

\begin{example}\label{ex:feng23-ex30}
	Consider the following example from~\cite[Example~30]{FengOOPSLA2023}, an example showing the limitation by the guard-strengthen principle:
	\[ \while{x \neq y}{ \pbranch{\assign{x}{y}}{1/3} {(\pbranch{\{\assign{z}{x};\ \assign{x}{y};\ \assign{y}{z}\}}{1/2}{\mathsf{diverge}})} } \]
	Here, we treat \( \mathsf{diverge} \) as an abbreviation for $\while{\mathbf{true}}{\skipstmt}$.

	The termination probability is given as the solution of the following equation system $E$
	\begin{gather*}
		X(x, y) =_{\mu} \ifexpr{x \neq y}{\tfrac{1}{3} X(y, y) + \tfrac{1}{3} X(y, x) + \tfrac{1}{3} Y()}{1}
		\qquad
		Y() =_{\mu} Y(),
	\end{gather*}
	which is \( 1 \)-bounded.
	Unfortunately, \( E \) has no ranking supermartingale since \( Y() =_\eta Y() \) actually diverges.
	Consider the system \( E' \) obtained by changing the weight of the diverging branch to $0$:
	\begin{gather*}
		X'(x, y) =_{\mu} \ifexpr{x \neq y}{\tfrac{1}{3} X'(y, y) + \tfrac{1}{3} X'(y, x) + 0 Y'()}{1}
		\qquad
		Y'() =_{\mu} 0
	\end{gather*}
	Obviously \( \interpret{E'} \le \interpret{E} \), and it is easy to show the ``termination'' of \( E' \) by giving a ranking supermartingale.
	Let \( \ell \) be a function given by \( \ell(x,y) = 1 \) when \( x = 1 \) and \( \ell(x,y) = 1/2 \) otherwise.
	Then \( [X' \mapsto \ell, Y' \mapsto 0] \) is an invariant of \( E' \), so it is a lower bound of the least solution of \( E \).
	This estimation is actually exact.
	\qed
\end{example}

\begin{example}\label{ex:biased-random-walk-ert-solved}
	Theorem~\ref{thm:unbounded-reasoning-principle}
	gives a lower bound for Example~\ref{ex:biased-random-walk-ert-definition} as follows.
	Let $E$ be the equation system~\eqref{eq:biased-random-walk-ert}.
	We first need to find $E'$ such that $\interpret{E'} \le \interpret{E}$, but in this case, we can take $E' = E$.
	We give a prefixed point \( u \) of $E$: let \( u(x) = \ifexpr{x > 0}{6 x}{0} \).
	A \( u \)-ranking supermartingale is given by \( r(x) = \ifexpr{x > 0}{9 x (x + 3)}{0} \).
	The check of \( \nonterm{\interpret{E}(r)} + u \le r \) is routine: it might be slightly easier by appealing to \( \nonterm{\interpret{E}} = \interpret{\nonterm{E}} \), where \( \nonterm{E} \) has the following equation over a quantitative predicate variable \( \nonterm{X} \).
	\[ \nonterm{X}(x : \mathbf{int}) \quad=_{\mu}\quad \ifexpr{x > 0}{\tfrac{2}{3} \nonterm{X}(x - 1) + \tfrac{1}{3} \nonterm{X}(x + 1)}{0} \]
	Then, \( \eta(x) = \ifexpr{x > 0}{b x}{0} \) (where \( 0 \le b \le 3 \)) is an invariant of \( E \) satisfying \( \eta \le u \).
	So it is a lower bound of \( \mu \interpret{E} \).
	This estimation is exact when \( b = 3 \).
	\qed
\end{example}

\begin{example}\label{ex:biased-random-walk-probability-using-u}
	Interestingly, a non-trivial \( u \) in Theorem~\ref{thm:unbounded-reasoning-principle} is also useful for reasoning about \(1\)-bounded equations.
	Let us again consider the termination probability of the biased random walk in Example~\ref{ex:biased-random-walk}.
	It is known that the least fixed point is \( \eta(x) = \ifexpr{x > 0}{(\tfrac{1}{2})^x}{1} \), and we aim to establish this result.
	In Example~\ref{ex:biased-random-walk-epsilon-subtraction}, this was shown by taking the limit of approximations.
	This exact bound can be given directly without taking the limit of approximations by using Theorem~\ref{thm:unbounded-reasoning-principle}.
	
	Let \( E \) be the fixed-point equation system described in Example~\ref{ex:biased-random-walk}.
	Let \( E' = E \) and \( u \) be the following function:
	\begin{equation*}
		u(x)
		\quad=\quad
		\ifexpr{x > 0}{\left(1 / 2\right)^x}{1}.
	\end{equation*}
	It is straightforward to verify that \( u \) is indeed a fixed-point and, in particular, a prefixed point.
	The invariant \( l \) is also \( u \).
	The remaining task is to find a \( u \)-ranking supermartingale.
	A \( u \)-ranking supermartingale can be explicitly given by
	\begin{equation*}
		r(x)
		\quad=\quad
		\ifexpr{x > 0}{18\left(2/3\right)^x}{1}.
	\end{equation*}
	Thus, by Theorem~\ref{thm:unbounded-reasoning-principle}, we conclude that \( u \le \mu\interpret{E} \).
	Since \( u \) itself is a fixed-point, it follows that \( u \) is indeed the least fixed-point.
	\qed
\end{example}

The exact bound in Example~\ref{ex:biased-random-walk-probability-using-u} cannot be obtained by existing methods that use stochastic invariants \cite{ChatterjeeCAV2022,AbateCAV2025,HenzingerCAV2025,FengOOPSLA2023}.
The common idea in these methods is to find a set \( I \subseteq D \) of program states such that \( \mathrm{Pr}(\text{termination} \mid \text{staying within $I$}) = 1 \) and \( \mathrm{Pr}(\text{staying within $I$}) \ge 1 - p \) for some \( p \).
Given such $I$, the termination probability is lower bounded as follows.
\[ \mathrm{Pr}(\text{termination}) \quad\ge\quad \mathrm{Pr}(\text{termination} \mid \text{staying within $I$}) \cdot \mathrm{Pr}(\text{staying within $I$}) \quad\ge\quad 1 - p \]
A pair \( (I, p) \) such that \( \mathrm{Pr}(\text{staying within $I$}) \ge 1 - p \) is called a \emph{stochastic invariant}.
Now, consider stochastic invariants for the biased random walk in Example~\ref{ex:biased-random-walk}.
Without loss of generality, we can assume that \( I = I[0, M] \) for some $M > 0$ where \( I[0, M] = \{ x \in \mathbb{N} \mid 0 \le x \le M \} \) because \( \mathrm{Pr}(\text{termination} \mid \text{staying within $I$}) = 1 \) must be satisfied.
However, calculations show \( \mathrm{Pr}(\text{staying within \(I[0, M]\)}) = 1 - (1 - 2^{-x}) / (1 - 2^{M + 1}) < 2^{-x} \) where $x$ is the initial state such that \( 0 < x \le M \).
Therefore, there is no stochastic invariant $(I, 1 - 2^{-x})$ that achieves the exact bound \( 2^{-x} \) for the biased random walk.

\begin{example}
	Consider the equation system in Example~\ref{ex:equation-systems-cost-moment} that describes the first and the second moment of the runtime of the biased random walk.
	The analysis for the first moment (Example~\ref{ex:biased-random-walk-ert-solved}) can be extended to the second moment as follows.
	Let \( u^{(1)}(x) = \ifexpr{x > 0}{6 x}{0} \), \( r^{(1)}(x) = \ifexpr{x > 0}{9 x (x + 3)}{0} \), and \( \eta^{(1)}(x) = \ifexpr{x > 0}{3 x}{0} \) be the witnesses taken from Example~\ref{ex:biased-random-walk-ert-solved}.
	We add the following functions as witnesses for the second moment.
	\begin{align*}
		u^{(2)}(x) &\quad=\quad \ifexpr{x > 0}{18 x^2 + 45 x}{0} \\
		r^{(2)}(x) &\quad=\quad \ifexpr{x > 0}{36 x^3 + (585/2) x^2 + (1683 / 2) x}{0} \\
		\eta^{(2)}(x) &\quad=\quad \ifexpr{x > 0}{9 x^2 + 24 x}{0} \tag*{\qed}
	\end{align*}
\end{example}

\subsection{On Completeness}
Our proof rule is not complete in the strict sense, i.e., cannot prove the exact lower bound for every equation system.
For example, our approach is unable to prove that the termination probability of the unbiased random walk (of dimension \( 1 \)) is \( 1 \).
Let \( E \) be the equation describing the termination of the unbiased random walk.
Since \( \probone > \nu \interpret{E'} \) for any strict under-approximation \( E' < E \), we have to choose \( E' = E \).
Any prefixed point \( u \) of \( E \) satisfies \( u \ge \probone \) as \( \probone \) is the least fixed-point, so a \( u \)-ranking supermartingale of \( E \) is a ranking submartingale for the unbiased random walk, but such a ranking submartingale does not exist (because the unbiased random walk does not exhibit positive almost-sure termination).

However, for \( 1 \)-bounded equation systems \( E \), our proof rule has completeness in the following weak sense: it can prove a lower bound arbitrary close to the exact bound.
Recall Example~\ref{ex:biased-random-walk-epsilon-subtraction}.
Note that there is no specific reason to subtract $\epsilon$ from the right branch of the probabilistic branching.
We could instead subtract $\epsilon$ from the left branch or multiply $0 \le \gamma < 1$ to both branches.
In any case, the program becomes almost surely terminating after these modifications.
This approach is called \emph{\( \gamma \)-scaling} and studied in \citet{UrabeLICS2017} and \citet{TakisakaATVA2018}.
Using \( \gamma \)-scaling, we can show the weak form\footnote{Note that the completeness results for lower bound estimation in the literature (e.g.~\citet{FengOOPSLA2023} and \citet{MajumdarPOPL2025}) are of this weak form.} of completeness of the reasoning principle.
\begin{theorem}
	If \( E \) is a \( 1 \)-bounded equation system (i.e.\ \( \interpret{E}(\probone) \le \probone \)), then the rule in Theorem~\ref{thm:unbounded-reasoning-principle} is complete in the following sense: the rule provides a lower bound arbitrary close to the exact bound (i.e., given \( d \in D \) and \( \epsilon > 0 \), there exists \( \eta' \) such that \( \eta' \le \mu \interpret{E}_{\le \probone} \) is provable by the rule in Theorem~\ref{thm:unbounded-reasoning-principle} and that \( \mu\interpret{E}_{\le \probone}(d) - \epsilon < \eta'(d) \le \mu \interpret{E}_{\le \probone}(d) \)).
	\qed
\end{theorem}
\begin{proof}
	We construct a witness for the rule in Theorem~\ref{thm:unbounded-reasoning-principle} as follows.
	Let \( K \coloneqq \interpret{E} \).
	Given \( 0 \le \gamma < 1 \), we define \( K_{\gamma} : \mathbb{E}(D) \to \mathbb{E}(D) \) by \( K_{\gamma}(X) := \gamma \cdot K(X) \).
	Obviously, $K_{\gamma}$ is an under-approximation of $K$.
	Let \( u(x) = 1 \) and \( r(x) = 1 / (1 - \gamma) \) for all \( x \in D \).
	Then, \( u \) is a prefixed point of \( K_{\gamma} \) and \( r \) is a \( u \)-ranking supermartingale because
	\[ \nonterm{K_{\gamma}}(r) + u \quad\le\quad \gamma \cdot K(1 / (1 - \gamma) \cdot \probone) + \probone \quad\le\quad 1 / (1 - \gamma) \cdot \probone \quad=\quad r. \]
	In this situation, the least fixed point $\eta' = \mu K_{\gamma}$ obviously gives an invariant.
	It remains to show that the lower bound \( \eta' = \mu K_{\gamma} \) approaches the least fixed point \( \mu K \) as \( \gamma \to 1 \).
	This is shown as follows.
	For any \( n \), we have \( (K_{\gamma})^n(\probzero) \ge \gamma^n K^n(\probzero) \) because \( K(\gamma \cdot X) \ge \gamma \cdot K(X) \) for any \( X \).
	\[ K(\gamma \cdot X) = \nonterm{K}(\gamma \cdot X) + K(\probzero)
		= \gamma \cdot \nonterm{K}(X) + K(\probzero)
		\ge \gamma \cdot \nonterm{K}(X) + \gamma \cdot K(\probzero)
		= \gamma \cdot K(X) \]
	Given \( \epsilon > 0 \) and \( d \in D \), we can choose sufficiently large \( n \) such that \( K^{n}(\probzero)(d) \ge \mu K(d) - \epsilon / 2 \) and sufficiently large \( \gamma \) such that $\gamma^n \ge 1 - \epsilon / 2$.
	Then, we have \( \mu K(d) - \epsilon \le \mu K_{\gamma}(d) = \eta'(d) \).
	\[ \mu K(d) - \epsilon \le K^{n}(\probzero)(d) - \epsilon / 2
		\le (1 - \epsilon / 2) \cdot K^{n}(\probzero)(d)
		\le \gamma^n \cdot K^{n}(\probzero)(d)
		\le (K_{\gamma})^{n}(\probzero)(d)
		\le \mu K_{\gamma}(d)
		\qedhere \]
\end{proof}

\section{Implementation and Experiments}\label{sec:implementation}

In this section, we describe a template-based algorithm to find lower-bound certificates and present experimental results.
To avoid incorrect results caused by numerical errors, our implementation focuses on bound checking problems that can be handled using SMT-based solvers.
A bound checking problem is a decision problem that asks whether the least fixed point of an equation system is lower bounded by a given value.
We do not consider bound inference algorithms that rely on numerical optimization.
However, if numerical errors are not a concern in a particular application, our reduction from bound checking problems to constraint solving problems can be used for bound inference as well.

We consider equation systems with multiple equations in this section.
In the theoretical development so far, we focused mainly on equation systems consisting of a single equation without loss of generality (Remark~\ref{rem:simplify-expectation}).
This is just for theoretical simplicity, and for implementation, it is more convenient to consider equation systems with multiple equations.

\paragraph{Queried equation system}
A bound checking problem is formalized as an equation system with a query.
A \emph{queried equation system} is a pair $(E, F \bowtie t)$ of an equation system $E = \{ X_i(\widetilde{x_i}) =_{\mu} F_i \mid i = 1, \dots, n \}$ and a query $F \bowtie t$ where $F$ is a quantitative formula over $\{ X_1, \dots, X_n \}$, $t \in \mathit{Exp}_{[0, \infty)}$ is a non-negative real term, and ${\bowtie} \in \{ {\le}, {\ge} \}$.
A queried equation system is valid if the solution of $E$ satisfies the query $F \bowtie t$, i.e, if $(X_1, \dots, X_n) = (X_1^{*}, \dots, X_n^{*})$ is the least fixed point of $E$, then $F[X_1^{*}/X_1, \dots, X_n^{*}/X_n]$ has a lower bound $F[X_1^{*}/X_1, \dots, X_n^{*}/X_n] \ge t$.
Since we are interested in the lower bound of the least fixed point, we mainly consider the case where $F \ge t$.
\begin{example}
	Consider the expected cost of the biased random walk in Example~\ref{ex:biased-random-walk-ert-definition}.
	Suppose that the aim is to verify that the expected runtime starting from $x = 1$ is lower bounded by $3$.
	Then, the queried equation system for this problem is $(E, X(1) \ge 3)$ where $E = \{ X(x : \mathbf{int}) =_{\mu} \ifexpr{x > 0}{\frac{2}{3} X(x - 1) + \frac{1}{3} X(x + 1) + 1}{0} \}$.
	\qed
\end{example}

Bound checking problems can be solved by finding a witness in Theorem~\ref{thm:unbounded-reasoning-principle}.
\begin{corollary}
	Let $(E, F \bowtie t)$ be a queried equation system where $E = \{ X_i(\widetilde{x_i}) =_{\mu} F_i \mid i = 1, \dots, n \}$.
	If there exists a witness $(E', u, r, \eta)$ that satisfies the conditions in Theorem~\ref{thm:unbounded-reasoning-principle} as well as $F[\eta_1/X_1, \dots, \eta_n/X_n] \bowtie t$, then the queried equation system $(E, F \bowtie t)$ is valid.
	\qed
\end{corollary}

\subsection{Reduction to Polynomial Quantified Entailments}\label{sec:reduction-to-pqe}

\subsubsection{Polynomial quantified entailments}
We aim to obtain a lower-bound certificate as a solution of polynomial quantified entailments (PQEs), which can be solved by the \textsc{PolyQEnt} PQE solver~\cite{ChatterjeeATVA2025}.
Before describing the reduction to PQEs, we first review the definition of PQEs.
A \emph{PQE} is a constraint of the following form.
\[ \forall \widetilde{x},\quad \Phi(\widetilde{x}; \widetilde{\theta}) \implies \Psi(\widetilde{x}; \widetilde{\theta}) \]
Here, $\Phi(\widetilde{x}; \widetilde{\theta})$ and $\Psi(\widetilde{x}; \widetilde{\theta})$ are boolean combinations of polynomial inequalities with unknown parameters $\widetilde{\theta}$.
A set of PQEs is \emph{satisfiable} if there exists an assignment for the unknown parameters that makes all the PQEs true.

PolyQEnt supports a few positivity theorems to solve PQEs.
Among them, we use Handelman's theorem.
Handelman's theorem is applicable when $\Phi(\widetilde{x}; \widetilde{\theta})$ is a conjunction of linear inequalities (i.e.\ \emph{polyhedra}), and $\Psi(\widetilde{x}; \widetilde{\theta})$ is a polynomial inequality.
Completeness is guaranteed when $\Phi(\widetilde{x}; \widetilde{\theta})$ is bounded, each linear inequality in $\Phi(\widetilde{x}; \widetilde{\theta})$ is non-strict, and $\Psi(\widetilde{x}; \widetilde{\theta})$ is a strict polynomial inequality.

\subsubsection{Input format}\label{sec:input-format-for-implementation}
The input of our implementation is a queried equation system $(E, F \ge t)$ where $E = \{ X_i(\widetilde{x_i}) =_{\mu} F_i \mid i = 1, \dots, n \}$.
We assume that (i) the domain of each quantitative predicate variable $X_i$ is $\mathbb{R}^m_i$, (ii) real or non-negative real terms are polynomials, and (iii) boolean terms are boolean combinations of linear inequalities.
For convenience, we further assume that quantitative formulas in $F_1, \dots, F_n, F$ are given in the \emph{normal form}, which is defined as follows.
\begin{equation}
	F_{\text{normal}} \quad\coloneqq\quad \sum_{j = 1}^{j_{\max}} [\varphi_j] \cdot A_j \qquad \text{where $\varphi_1, \dots, \varphi_{j_{\max}}$ are polyhedra and mutually disjoint.}
	\label{eq:normal-form-quantitative-formula}
\end{equation}
Here, $[\varphi_j]$ is the Iverson bracket, and $A$ is a quantitative formula of the following form.
\[ A \quad\coloneqq\quad \sum_{k = 1}^{k_{\max}} t_k \cdot X_{i_{k}}(\widetilde{e}_{k}) + t_0 \qquad\qquad \text{where $t_0, \dots, t_{k_{\max}}$ and $\widetilde{e}_{1}, \dots \widetilde{e}_{k_{\max}}$ are polynomial.} \]
Converting a quantitative formula to the normal form is straightforward.
\begin{example}
	Consider the quantitative formula \( \ifexpr{x > 0}{\frac{1}{3} X(x - 1) + \frac{2}{3} X(x + 1)}{1} \).
	The normal form is given by $[x > 0] \cdot \big(\frac{1}{3} X(x - 1) + \frac{2}{3} X(x + 1)\big) + [x \le 0] \cdot 1$.
	\qed
\end{example}

\subsubsection{Reduction}
Given a pair of a queried equation system $(E, F \ge t)$, we aim to find a witness in Theorem~\ref{thm:unbounded-reasoning-principle}.
We first explain our reduction for the case where $E'$ is fixed to $E' = E$ and later extend it to allow $E' \neq E$.
When $E' = E$, we need to find three witnesses for a lower bound: a prefixed point $u = (u_1, \dots, u_n)$, a $u$-ranking supermartingale $r = (r_1, \dots, r_n)$, and an invariant $\eta = (\eta_1, \dots, \eta_n)$ where $n$ is the number of equations in $E$.
We consider polynomial templates for these witnesses.
We impose the constraints that the witnesses satisfy (a) the conditions in Theorem~\ref{thm:unbounded-reasoning-principle}, (b) non-negativity, and (c) the query $F \ge t$.
\begin{align}
	u_i &\ge F_i[u/X] &
	r_i &\ge u_i + (\nonterm{F})_i[r/X] &
	\eta_i &\le u_i &
	\eta_i &\le F_i[\eta/X]
	\label{eq:template-constraint-2} \\
	u_i &\ge 0 &
	r_i &\ge 0 &
	\eta_i &\ge 0 &
	F[\eta/X] &\ge t
	\label{eq:template-constraint-3}
\end{align}
Here, $F[u/X]$ is a shorthand for $F[u_1/X_1, \dots, u_n/X_n]$ and $\nonterm{F}$ is obtained by replacing $t_j$ in~\eqref{eq:normal-form-quantitative-formula} with $0$.
Now, we can decompose these constraints into a set of PQEs.
For example, consider the constraint $u_i \ge F_i[u/X]$.
Since we assume that $F_i$ is in the normal form~\eqref{eq:normal-form-quantitative-formula}, we can write $F_i[u/X]$ as $\sum_{j = 1}^{j_{\max}} [\varphi_j] \cdot A_j[u/X]$.
Therefore, the constraint $u_i \ge F_i[u/X]$ is equivalent to (the conjunction of) PQEs $\{ \forall \widetilde{x_i}, \varphi_j \implies u_i \ge A_j[u/X] \mid j = 1, \dots, j_{\max} \}$.
This reduction itself is obviously sound and complete.
Moreover, the resulting PQEs are in the scope of Handelman's theorem and can be solved by \textsc{PolyQEnt}.

\begin{example}
	Consider the expected runtime of the biased random walk in Example~\ref{ex:biased-random-walk-ert-definition} (see also Example~\ref{ex:biased-random-walk-ert-solved}).
	To show that the expected cost from $x = 1$ is lower bounded by $3$, we need to find a witness $(u, r, \eta)$ satisfying the following constraints.
	\begin{align*}
		&u(x) \quad\ge\quad \ifexpr{x > 0}{\tfrac{2}{3} u(x - 1) + \tfrac{1}{3} u(x + 1) + 1}{0} \\
		&r(x) \quad\ge\quad u(x) + \ifexpr{x > 0}{\tfrac{2}{3} r(x - 1) + \tfrac{1}{3} r(x + 1)}{0} \\
		&\eta(x) \quad\le\quad \ifexpr{x > 0}{\tfrac{2}{3} \eta(x - 1) + \tfrac{1}{3} \eta(x + 1) + 1}{0} \\
		&\eta(x) \le u(x) \qquad
		u(x) \ge 0 \qquad
		r(x) \ge 0 \qquad
		\eta(x) \ge 0 \qquad
		\eta(1) \ge 3
	\end{align*}
	Assuming that $u$, $r$, and $\eta$ are polynomials, we can decompose these constraints into a set of PQEs.
	\begin{align*}
		&\forall x, x > 0 \implies u(x) \ge \tfrac{2}{3} u(x - 1) + \tfrac{1}{3} u(x + 1) + 1 &&
		\forall x, x \le 0 \implies u(x) \ge 0 \\
		&\forall x, x > 0 \implies r(x) \ge u(x) + \tfrac{2}{3} r(x - 1) + \tfrac{1}{3} r(x + 1) &&
		\forall x, x \le 0 \implies r(x) \ge u(x) \\
		&\forall x, x > 0 \implies \eta(x) \le \tfrac{2}{3} \eta(x - 1) + \tfrac{1}{3} \eta(x + 1) + 1 &&
		\forall x, x \le 0 \implies \eta(x) \le 0 \\
		&\forall x, u(x) \ge 0 \qquad
		\forall x, r(x) \ge 0 \qquad
		\forall x, \eta(x) \ge 0 \qquad
		\eta(1) \ge 3
		\tag*{\qed}
	\end{align*}
\end{example}

\subsection{Minor Extensions}\label{sec:minor-extensions}

\subsubsection{Piecewise polynomials}\label{sec:piecewise-polynomial-template}
We extend the polynomial template for $u$, $r$, and $\eta$ to piecewise polynomials.
This is straightforward once we decide how to divide the domain to pieces.
As a heuristic, we use the same branching structure as the equation system: for each equation $X_i(\widetilde{x}) =_{\mu} F_i$ in $E$ with the normal form of $F_i$ given as $\sum_{j = 1}^{j_{\max}} [\varphi_j] \cdot A_j$, we define the template for $u_i$, $r_i$, and $\eta_i$ by
\begin{equation}
	u_i(\widetilde{x}),\ r_i(\widetilde{x}),\ \eta_i(\widetilde{x}) \qquad=\qquad \sum_{j = 1}^{j_{\max}} [\varphi_j] \cdot p_j(\widetilde{x})
	\label{eq:template-expectations}
\end{equation}
where $p_j(\widetilde{x})$ is a polynomial with unknown coefficients.
Some existing template-based implementations such as \cite{ChatterjeeCAV2022} assumes that invariants, which in this case means a set of states $I \subseteq D$ such that a program stays in $I$ during its execution, are either manually provided or precomputed by existing invariant synthesis tools for non-probabilistic programs.
In contrast, our current implementation does not use manually provided invariants or precomputed invariants.
Therefore, we often need piecewise polynomial templates to satisfy non-negativity constraints.

\begin{example}
	Consider the expected runtime of the biased random walk in Example~\ref{ex:biased-random-walk-ert-definition}.
	Then, using the same branching structure as the equation system, the piecewise polynomial template for $u$, $r$, and $\eta$ is given by
	\[ u(x), r(x), \eta(x) \quad=\quad \ifexpr{x > 0}{p_1(x)}{p_2(x)} \]
	where $p_1$ and $p_2$ are polynomials with unknown coefficients.
	\qed
\end{example}

\subsubsection{Template for equation systems}\label{sec:template-equation-system}
We define the template for $E'$ by applying guard-strengthening and weight-subtraction (Section~\ref{sec:1ufp-reasoning}) to the original equation system $E$.
Specifically, for each equation $X(\widetilde{x}) =_{\mu} F_i$ in $E$, we define the corresponding equation $X(\widetilde{x}) =_{\mu} F_i'$ in $E'$ by (a) strengthening $\varphi_j$ and (b) multiplying $t_{j,k}$ by unknown parameters $0 \le a_{j, k} \le 1$.
\[ F_i = \sum_{j = 1}^{j_{\max}} [\varphi_j] \cdot \sum_{k = 1}^{k_{\max, j}} t_{j, k} \cdot X_{i_{j, k}}(\widetilde{e}_{j, k}) + t_{j, 0} \quad\mapsto\quad F_i' = \sum_{j = 1}^{j_{\max}} [\varphi' \land \varphi_j] \cdot \sum_{k = 1}^{k_{\max, j}} a_{j, k} t_{j, k} \cdot X_{i_{j, k}}(\widetilde{e}_{j, k}) + a_{j, 0} t_{j, 0} \]
Here, we use the normal form~\eqref{eq:normal-form-quantitative-formula} for $F_i$, and $\varphi'$ is a conjunction of linear inequalities with unknown coefficients.
If $0 \le a_{j, k} \le 1$ is satisfied, then the resulting equation system $E'$ satisfies $\interpret{E'} \le \interpret{E}$.
Therefore, we can extend the reduction to PQEs in Section~\ref{sec:reduction-to-pqe} to use $E'$ instead of $E$ and solve the PQEs together with the additional constraints $0 \le a_{j, k} \le 1$.

\subsection{Experimental Results}
\begin{table}[tbp]
	\caption{The results for bound checking problems.
	The Result column shows whether the implementation found witnesses for the benchmarks successfully (``valid'') or not (``unknown'', which means that \textsc{PolyQEnt} concluded that there is no solution within the template we consider).
	The timeout was set to 180 seconds.
	}
		\label{tab:results}
		\vspace{-1em}
		\begin{subtable}{\textwidth}
			\centering
			\subcaption{Benchmarks bound-checked against exact bounds.}
			\label{tab:results-exact-bounds}
			\small
			\begin{tabular}{c|ccc|ccc}
		& \multicolumn{3}{c}{Lower bound} & \multicolumn{3}{c}{Upper bound} \\
		Benchmark & Config & Result & Time (sec) & Config & Result & Time (sec) \\
		\hline
		\texttt{ert\_coin\_flip} & A & valid & 0.246 & A & valid & 0.268 \\
		\texttt{ert\_random\_walk} (Example~\ref{ex:biased-random-walk-ert-definition}) & A & valid & 0.593 & A & valid & 0.357 \\
		\hline
		\texttt{ert\_random\_walk\_2nd} (Example~\ref{ex:equation-systems-cost-moment}) & A & valid & 1.216 & A & valid & 0.481 \\
		\hline
		\texttt{feng23\_ex30} (Example~\ref{ex:feng23-ex30}) & A & valid & 0.916 & A & valid & 0.454 \\
		\texttt{chatterjee22\_fig6} & A & valid & 0.268 & A & valid & 0.425 \\
		\texttt{chatterjee22\_fig7} & A & valid & 0.289 & A & valid & 0.454 \\
		\texttt{chatterjee22\_fig8} & A & valid & 0.394 & A & valid & 0.599 \\
		\texttt{chatterjee22\_fig20} & A & valid & 0.541 & A & valid & 0.326 \\
		\texttt{chatterjee22\_fig21} & A & valid & 0.584 & A & valid & 0.321 \\
		\texttt{chatterjee22\_fig24} & C1 & valid & 8.591 & B & valid & 8.893 \\
		\texttt{hark19\_ex7} & C2 & valid & 0.347 & C3 & valid & 0.322 \\
		\texttt{hark19\_ex49} & A & valid & 1.467 & A & valid & 0.505 \\
		\hline
		\texttt{olmedo18\_intro\_cwp1} & A & valid & 0.417 & A & valid & 0.354 \\
		\texttt{olmedo18\_intro\_cwp2} & A & valid & 0.416 & A & valid & 0.326
	\end{tabular}
		\end{subtable}

	\vspace{1em}
	\begin{subtable}{\textwidth}
		\centering
		\subcaption{Benchmarks bound-checked against approximate values. ``?'' means that the exact value is unknown.}
		\label{tab:results-approximate-bounds}
		\small
	\begin{tabular}{c|c|cccc}
		& & \multicolumn{4}{c}{Lower bound} \\
		Benchmark & Exact & Bound & Config & Result & Time (sec) \\
		\hline
		\texttt{feng23\_ex19} & 2.0 & 1.0 & C2 & valid & 24.641 \\
		\texttt{chatterjee22\_fig14}\footnotemark & $? \in [0, 1]$ & 0.001 & B & valid & 14.235 \\
		\texttt{chatterjee22\_fig18} & $? \in [0, 1]$ & 0.49 & C1 & valid & 25.695 \\
		\texttt{chatterjee22\_fig22} & $? \in [0, 1]$ & 0.9 & C1 & valid & 0.873 \\
		\texttt{chatterjee22\_fig23} & $? \in [0, 1]$ & 0.92 & A & unknown & 15.046 \\
		\texttt{hark19\_ex53} & $\infty$ & $10000$ & C2 & valid & 0.290
	\end{tabular}
	\end{subtable}
\end{table}
\footnotetext{We manually provided hints for the branching structure of the templates for $u, r, \eta$ and $E'$.
Our tool usually construct the template by directly following the branching structure of the if statements in the given input.
For \texttt{chatterjee22\_fig14}, we manually refine the branching structure by splitting it into more fine-grained cases.
This corresponds to increasing the number of pieces in the piecewise polynomial template and therefore amounts to using a more expressive template.
}

We implemented a tool, \OurTool, which can solve queried equation systems for both lower and upper bounds of least fixed points.\footnote{Available at \OurToolUrl}
For lower bounds, \OurTool\ tries to find witnesses using the procedure described in Section~\ref{sec:reduction-to-pqe} and~\ref{sec:minor-extensions}.
For upper bounds, it applies the Knaster--Tarski theorem using the same templates as for lower bounds and then solves the resulting PQEs by \textsc{PolyQEnt}.

We collected all probabilistic programs from the literature about lower-bound verification~\cite{FengOOPSLA2023,HarkPOPL2020,ChatterjeeCAV2022}\footnote{Example numbers for \texttt{hark19\_ex*} are based on the arxiv version of \cite{HarkPOPL2020}.}, excluding benchmarks that require non-polynomial functions as witnesses.
We also excluded benchmarks that use continuous distributions, as our implementation does not support them.
Example~\ref{ex:equation-systems-cost-moment},~\ref{ex:biased-random-walk-ert-definition},~\ref{ex:biased-random-walk}, and a simple coin flip program are included in our benchmark set too.
The benchmarks cover a wide range of quantitative properties, including expected runtime (\texttt{ert\_coin\_flip}, \texttt{ert\_random\_walk}), higher moments of runtime (\texttt{ert\_random\_walk\_2nd}), weakest preexpectation (\texttt{feng23\_ex*}, \texttt{chatterjee22\_fig*}, \texttt{hark19\_ex*}), and conditional weakest preexpectation (\texttt{olmedo18\_intro\_cwp*}).
One of the benchmarks, \texttt{chatterjee22\_fig7} involves nondeterminism.
We do not have benchmarks for the conditional weakest preexpectation, as programs in \cite{OlmedoTOPLAS2018} require exponential functions as witnesses.
We show an example of such exponential witnesses for the conditional weakest preexpectation in \referappendix{Example}{C.1}{ex:olmedo18}.
We translated these programs into the input format, queried equation systems.
Translating probabilistic programs into equation systems is straightforward, as explained in Section~\ref{sec:problem}.
Moreover, we added queries that specify desired lower/upper bounds.
Most of the bounds used here are exact bounds except for \texttt{feng23\_ex19} and \texttt{hark19\_ex53}. For \texttt{feng23\_ex19}, \textsc{PolyQEnt} timed out when verifying the exact lower bound. Therefore, we report the result using an approximate lower bound along with the exact upper bound. The reason for using approximate bounds in \texttt{hark19\_ex53} is that the exact bound is $\infty$ in this case. As the exact bounds for \texttt{chatterjee22\_fig\{14,18,22,23\}} are unknown, we used the lower bounds reported in \cite{ChatterjeeCAV2022} as approximate lower bounds and omitted upper bound verification.

We were unable to compare our approach with existing tools for lower-bound verification \cite{ChatterjeeCAV2022,FengOOPSLA2023,HarkPOPL2020,TakisakaATVA2018} for several reasons.
First, differences in problem settings make a direct comparison difficult.
Our implementation is based on SMT solving, whereas some existing tools \cite{TakisakaATVA2018} rely on numerical optimization.
SMT-based approaches provide exact results without numerical error, while numerical-optimization-based approaches are typically more efficient but may introduce approximation errors.
Moreover, our implementation targets bound-checking problems, where the goal is to establish a lower bound exceeding a given threshold.
In contrast, some existing tools aim to infer tight lower bounds via optimization without assuming a predefined threshold.
Second, some prior works \cite{ChatterjeeCAV2022,FengOOPSLA2023} do not provide implementations or do not make them publicly available.
Finally, \textsc{cegispro2}, which implements the method of \cite{HarkPOPL2020}, imposes various restrictions on input programs.
Due to these restrictions, we were unable to apply \textsc{cegispro2} to most of our benchmarks.

Table~\ref{tab:results} shows the results of our experiments.
Our experiments were conducted on 12th Gen Intel(R) Core(TM) i7-1270P 2.20 GHz with 32 GB of memory, and the timeout was set to 180 seconds.
We used 5 configurations in the experiments.
Configuration A is the default configuration, which uses the algorithm described in Section~\ref{sec:reduction-to-pqe} with piecewise polynomial templates in Section~\ref{sec:piecewise-polynomial-template} but not with the template for equation systems in Section~\ref{sec:template-equation-system}.
The degrees of polynomial templates for $u$, $r$, and $\eta$ are 2, 3, and 2, respectively.
Configuration B differs from A in that the degrees of polynomial templates for $u$, $r$, and $\eta$ are 1, 1, and 1, respectively.
Configuration C1 differs from B by using the template for equation systems in Section~\ref{sec:template-equation-system}, where $\varphi'$ in the template for $E'$ is fixed to $\top$. Configuration C2 differs from C1 by allowing one linear inequality with unknown coefficients in $\varphi'$.
Configuration C3 is almost the same as Configuration C2, but we add an if-branching outside the template~\eqref{eq:template-expectations} for $u$, $r$, and $\eta$ where the boolean condition is given by a linear inequality with unknown coefficients.
Configurations used for each benchmark are shown in Table~\ref{tab:results}.

\paragraph{Strengths and limitations}
Our implementation was able to find \emph{exact} bounds for most of the benchmarks, which is a significant advantage over existing methods that can only find approximate bounds.
Our implementation also covered benchmarks that are beyond the scope of existing methods.
For example, \texttt{feng23\_ex30} is a benchmark that cannot be handled by the guard-strengthening method \cite{FengOOPSLA2023}.
Lower bounds of the expected runtime and higher moments of runtime are also beyond the scope of existing methods, but our implementation was able to find lower bounds.
On the other hand, our current implementation cannot handle benchmarks that require witnesses beyond the scope of polynomial templates.
We note that exponential templates are studied in the literature \cite{WangPLDI2021}.
However, this method relies on convex optimization and is not based on SMT solving.
Hence, we leave implementing exponential templates as future work.
Our implementation also does not support continuous distributions.
However, this is just a design choice.
If we change our implementation to take invariants (a set of states in which a program stays during its execution) as a part of the input and use (non-piecewise) polynomial templates, then it would be straightforward to support continuous distributions as well.

The failure to obtain an exact lower bound for \texttt{feng23\_ex19}, the need for manual hints in \texttt{chatterjee22\_fig14}, and the timeout in \texttt{chatterjee22\_fig23} may be improved by future enhancements to the backend solver, \textsc{PolyQEnt}. However, it is also true that the current implementation places a heavy burden on \textsc{PolyQEnt} by requiring it to simultaneously search for the four witnesses in Theorem~\ref{thm:unbounded-reasoning-principle}. We leave improvements to this part as future work.

\section{Related Work}
\label{sec:related}

One of the most important quantitative properties of probabilistic programs is the termination probability.
The verification of almost sure termination is a special case where the lower bound of the termination probability is $1$, which is sometimes called \emph{qualitative} termination analysis.
There is a line of work on verifying almost sure termination of probabilistic programs via ranking supermartingales~\cite{ChakarovCAV2013,FerrerFioritiPOPL2015,ChatterjeeCAV2016,TakisakaATVA2018,Kenyon-RobertsLICS2021,AgrawalPOPL2018,TakisakaCAV2024}.
Some of these works have also implemented template-based reduction to constraint solving problems such as linear programming or semidefinite programming.
\emph{Quantitative} termination analysis, where lower bounds can be any $p \in [0, 1]$, has been studied in~\cite{ChatterjeePOPL2017,ChatterjeeCAV2022,MajumdarPOPL2025}, where the notion of \emph{stochastic invariants} play an important role.
We note that using under-approximation \( E' \) in Theorem~\ref{thm:unique-fixed-point-by-ranking-supermartingale} allows us to simulate the use of stochastic invariants, as we describe in \referappendix{Appendix}{D.3}{sec:simulating-stochastic-invariants}.

A more general problem, lower bounds of the weakest preexpectation, has been studied in~\cite{McIver2005,HarkPOPL2020,FengOOPSLA2023}.
It was pointed out that the notion of uniform integrability is closely related to lower bounds of least fixed points~\cite[Theorem 10]{HarkPOPL2020}.
Uniform integrability gives a necessary and sufficient condition for the lower bounds, but it is not easy to check in general.
A few sufficient criteria for uniform integrability have been proposed in~\cite{HarkPOPL2020}, but these criteria are restricted to programs that are almost surely terminating.
Our Theorem~\ref{thm:unique-fixed-point-by-ranking-supermartingale} can be viewed as a new sufficient condition for uniform integrability \cite[Definition 11]{HarkPOPL2020}.
Proof principles for lower bounds have been also proposed in~\cite{McIver2005} but are restricted to bounded expectations.
To solve these limitations, the guard-strengthening technique has been proposed in~\cite{FengOOPSLA2023}.
However, there are cases where guard-strengthening is not applicable (e.g.~\cite[Example~30]{FengOOPSLA2023}).

Lower bounds of expected runtime have been studied in~\cite{HarkPOPL2020,FuVMCAI2019,WangPLDI2021a,WangPLDI2019,ChatterjeeOOPSLA2024}.
These methods are based on optional stopping theorems, which typically require some difference-boundedness conditions, such as difference-boundedness~\cite{FuVMCAI2019}, conditional difference-boundedness~\cite{HarkPOPL2020}, and bounded-update property~\cite{WangPLDI2021a,WangPLDI2019}.
An exception is the method proposed in~\cite{ChatterjeeOOPSLA2024}\footnote{This work mainly focuses on giving upper bounds. However since it allows general cost, including negative values, it can also be applied to lower-bound analysis by negating the cost.}, which instead requires the lower-bounded total cost condition.
Some of these works~\cite{WangPLDI2019,WangPLDI2021a,ChatterjeeOOPSLA2024} allows both positive and negative costs.
Although our setting does not allow negative expectations, we can handle negative costs to some extent, as we describe in \referappendix{Appendix}{G}{sec:experiment-negative-cost}.

The work~\cite{WangPLDI2024} studies lower bounds of expected weight of probabilistic programs with soft conditioning.
It proposes an OST-based method, which requires the bounded-update property, as well as a method based on the uniqueness of fixed points, which requires boundedness of expectations and almost-sure termination.

As far as we know, the only existing methods that apply the uniqueness of fixed points to probabilistic programs is $\gamma$-scaling submartingales~\cite{TakisakaATVA2018,UrabeLICS2017}, which give lower bounds of the termination probability, and the approaches in~\cite{WangPLDI2021,WangPLDI2024} rely on the almost-sure termination.
For the former, the idea of $\gamma$-scaling was obtained through abstract study of fixed points using category theory.
The key idea of their framework is to give the set of truth values (in this case, $[0, 1]$) a mathematical structure that guarantees the uniqueness.
Their approach can be seen as the ``dual'' of ours: while they refine the notion of truth values to obtain a unique fixed point, we modify the verification target to achieve the same result.
It is also unclear how we can give such a structure to the unbounded case (i.e.\ $[0, \infty]$).
The latter approaches~\cite{WangPLDI2021,WangPLDI2024} are closer to ours, but they are restricted to the bounded setting and do not give a general reasoning principle for lower bounds of unbounded expectations.

\section{Conclusions and Future Work}\label{sec:conc}
We proposed a new method for obtaining lower bounds of the least fixed point.
Our method is uniformly applicable to a wide range of quantitative properties of probabilistic programs, including the weakest preexpectation (with conditioning), the expected runtime, and higher moments of the runtime.
Technically, we give a new sufficient condition for the uniqueness of fixed points by generalising ranking supermartingales.
This result led to a new reasoning principle for lower bounds, which is applicable to quantitative verification problems that were beyond the scope of existing methods.
We also implemented our method and showed the effectiveness of our method through experiments.
As future work, we aim to apply exponential templates \cite{WangPLDI2021} to our framework and to implement it.
Handling negative costs is another interesting future work (\referappendix{Appendix}{G}{sec:experiment-negative-cost}).

\section*{Data Availability Statement}
\OurTool\ source code can be found at \OurToolUrl. We also provide an artifact \cite{artifact} to reproduce the evaluation in Section~\ref{sec:implementation}, including a snapshot of \OurTool, docker images, and benchmarks.

\begin{acks}
We thank the anonymous reviewers for their helpful comments and suggestions.
This study was supported by JST K Program Grant Number JPMJKP24U2; JSPS KAKENHI Grant Number JP23K24820, JP25H00446, JP25K21183, JP25K24739.
\end{acks}

\bibliographystyle{ACM-Reference-Format}
\bibliography{lower-bound,abbrv,prog_lang,arxiv}

@string{PLDI2006 = {PLDI '06}}

@string{SAS2005 = {SAS '05}}

@string{LNCS = {LNCS}}

@string{PACMPL = {Proceedings of the ACM on Programming Languages}}

@string{COMPUTER = {IEEE Computer}}

@STRING{ACM = {ACM}}

@STRING{IEEE = {IEEE}}

@STRING{Springer = {Springer}}

@misc{arxiv,
  title         = {Supermartingales for Unique Fixed Points: A Unified Approach to Lower Bound Verification},
  author        = {Satoshi Kura and Hiroshi Unno and Takeshi Tsukada},
  year          = {2025},
  eprint        = {2504.04132},
  archiveprefix = {arXiv},
  primaryclass  = {cs.LO},
  url           = {https://arxiv.org/abs/2504.04132}
}

@software{artifact,
  author    = {Kura, Satoshi and
               Unno, Hiroshi and
               Tsukada, Takeshi},
  title     = {Supermartingales for Unique Fixed Points: A
               Unified Approach to Lower Bound Verification --
               Artifact
               },
  month     = apr,
  year      = 2026,
  publisher = {Zenodo},
  doi       = {10.5281/zenodo.19394993},
  url       = {https://doi.org/10.5281/zenodo.19394993}
}

@inproceedings{AbateCAV2025,
  title = {Quantitative Supermartingale Certificates},
  booktitle = {Computer Aided Verification},
  author = {Abate, Alessandro and Giacobbe, Mirco and Roy, Diptarko},
  year = 2025,
  pages = {3--28},
  publisher = {Springer Nature Switzerland},
  address = {Cham},
  doi = {10.1007/978-3-031-98679-6_1},
  isbn = {978-3-031-98679-6}
}

@article{AgrawalPOPL2018,
  title = {Lexicographic Ranking Supermartingales: An Efficient Approach to Termination of Probabilistic Programs},
  shorttitle = {Lexicographic Ranking Supermartingales},
  author = {Agrawal, Sheshansh and Chatterjee, Krishnendu and Novotn{\'y}, Petr},
  year = 2018,
  month = jan,
  journal = {Proceedings of the ACM on Programming Languages},
  volume = {2},
  number = {POPL},
  pages = {1--32},
  issn = {2475-1421, 2475-1421},
  doi = {10.1145/3158122},
  url = {http://dl.acm.org/doi/10.1145/3158122},
  urldate = {2020-02-25},
  langid = {english}
}

@article{AguirreMSCS2022,
  title = {Weakest Preconditions in Fibrations},
  author = {Aguirre, Alejandro and Katsumata, Shin{-}ya and Kura, Satoshi},
  year = 2022,
  month = oct,
  journal = {Mathematical Structures in Computer Science},
  volume = {32},
  number = {4},
  pages = {472--510},
  issn = {0960-1295, 1469-8072},
  doi = {10.1017/S0960129522000330},
  url = {https://www.cambridge.org/core/product/identifier/S0960129522000330/type/journal_article},
  urldate = {2022-10-29},
  langid = {english}
}

@inproceedings{BatzTACAS2023,
  title = {Probabilistic {{Program Verification}} via {{Inductive Synthesis}} of {{Inductive Invariants}}},
  booktitle = {Tools and {{Algorithms}} for the {{Construction}} and {{Analysis}} of {{Systems}}},
  author = {Batz, Kevin and Chen, Mingshuai and Junges, Sebastian and Kaminski, Benjamin Lucien and Katoen, Joost-Pieter and Matheja, Christoph},
  year = 2023,
  series = {Lecture {{Notes}} in {{Computer Science}}},
  volume = {13994},
  pages = {410--429},
  publisher = {Springer Nature Switzerland},
  address = {Cham},
  doi = {10.1007/978-3-031-30820-8_25},
  url = {https://link.springer.com/10.1007/978-3-031-30820-8_25},
  urldate = {2024-02-29},
  isbn = {978-3-031-30819-2 978-3-031-30820-8},
  langid = {english}
}

@inproceedings{BeutnerPLDI2021,
  title = {On Probabilistic Termination of Functional Programs with Continuous Distributions},
  booktitle = {Proceedings of the 42nd {{ACM SIGPLAN International Conference}} on {{Programming Language Design}} and {{Implementation}}},
  author = {Beutner, Raven and Ong, Luke},
  year = 2021,
  month = jun,
  pages = {1312--1326},
  publisher = {ACM},
  address = {Virtual Canada},
  doi = {10.1145/3453483.3454111},
  url = {https://dl.acm.org/doi/10.1145/3453483.3454111},
  urldate = {2023-09-19},
  isbn = {978-1-4503-8391-2},
  langid = {english}
}

@inproceedings{ChakarovCAV2013,
  title = {Probabilistic {{Program Analysis}} with {{Martingales}}},
  booktitle = {Computer {{Aided Verification}}},
  author = {Chakarov, Aleksandar and Sankaranarayanan, Sriram},
  year = 2013,
  series = {Lecture {{Notes}} in {{Computer Science}}},
  volume = {8044},
  pages = {511--526},
  publisher = {Springer Berlin Heidelberg},
  address = {Saint Petersburg, Russia},
  doi = {10.1007/978-3-642-39799-8_34},
  url = {http://link.springer.com/10.1007/978-3-642-39799-8_34},
  urldate = {2020-03-02},
  isbn = {978-3-642-39798-1 978-3-642-39799-8}
}

@inproceedings{ChatterjeeATVA2025,
  title = {{{PolyQEnt}}: {{A Polynomial Quantified Entailment Solver}}},
  shorttitle = {{{PolyQEnt}}},
  booktitle = {Automated {{Technology}} for {{Verification}} and {{Analysis}}},
  author = {Chatterjee, Krishnendu and Goharshady, Amir Kafshdar and Goharshady, Ehsan Kafshdar and Karrabi, Mehrdad and Saadat, Milad and Seeliger, Maximilian and {\v Z}ikeli{\'c}, {\DJ}or{\dj}e},
  year = 2025,
  series = {Lecture {{Notes}} in {{Computer Science}}},
  volume = {16145},
  pages = {411--424},
  publisher = {Springer Nature Switzerland},
  address = {Cham},
  doi = {10.1007/978-3-032-08707-2_19},
  url = {https://link.springer.com/10.1007/978-3-032-08707-2_19},
  urldate = {2025-11-12},
  isbn = {978-3-032-08706-5 978-3-032-08707-2},
  langid = {english}
}

@inproceedings{ChatterjeeCAV2016,
  title = {Termination {{Analysis}} of {{Probabilistic Programs Through Positivstellensatz}}'s},
  booktitle = {Computer {{Aided Verification}}},
  author = {Chatterjee, Krishnendu and Fu, Hongfei and Goharshady, Amir Kafshdar},
  year = 2016,
  series = {Lecture {{Notes}} in {{Computer Science}}},
  volume = {9779},
  pages = {3--22},
  publisher = {Springer International Publishing},
  address = {Cham},
  doi = {10.1007/978-3-319-41528-4_1},
  url = {http://link.springer.com/10.1007/978-3-319-41528-4_1},
  urldate = {2020-03-02},
  isbn = {978-3-319-41527-7 978-3-319-41528-4},
  langid = {english}
}

@inproceedings{ChatterjeeCAV2022,
  title = {Sound and {{Complete Certificates}} for {{Quantitative Termination Analysis}} of {{Probabilistic Programs}}},
  booktitle = {Computer {{Aided Verification}}},
  author = {Chatterjee, Krishnendu and Goharshady, Amir Kafshdar and Meggendorfer, Tobias and {\v Z}ikeli{\'c}, {\DJ}or{\dj}e},
  year = 2022,
  series = {Lecture {{Notes}} in {{Computer Science}}},
  volume = {13371},
  pages = {55--78},
  publisher = {Springer International Publishing},
  address = {Cham},
  doi = {10.1007/978-3-031-13185-1_4},
  url = {https://link.springer.com/10.1007/978-3-031-13185-1_4},
  urldate = {2022-08-08},
  isbn = {978-3-031-13184-4 978-3-031-13185-1},
  langid = {english}
}

@article{ChatterjeeOOPSLA2024,
  title = {Quantitative {{Bounds}} on {{Resource Usage}} of {{Probabilistic Programs}}},
  author = {Chatterjee, Krishnendu and Goharshady, Amir Kafshdar and Meggendorfer, Tobias and {\v Z}ikeli{\'c}, {\DJ}or{\dj}e},
  year = 2024,
  month = apr,
  journal = {Proceedings of the ACM on Programming Languages},
  volume = {8},
  number = {OOPSLA1},
  pages = {362--391},
  issn = {2475-1421},
  doi = {10.1145/3649824},
  url = {https://dl.acm.org/doi/10.1145/3649824},
  urldate = {2026-03-09},
  langid = {english}
}

@inproceedings{ChatterjeePOPL2017,
  title = {Stochastic Invariants for Probabilistic Termination},
  booktitle = {Proceedings of the 44th {{ACM SIGPLAN Symposium}} on {{Principles}} of {{Programming Languages}} - {{POPL}} 2017},
  author = {Chatterjee, Krishnendu and Novotn{\'y}, Petr and {\v Z}ikeli{\'c}, {\DH}or{\dj}e},
  year = 2017,
  pages = {145--160},
  publisher = {ACM Press},
  address = {Paris, France},
  doi = {10.1145/3009837.3009873},
  url = {http://dl.acm.org/citation.cfm?doid=3009837.3009873},
  urldate = {2020-03-14},
  isbn = {978-1-4503-4660-3},
  langid = {english}
}

@article{DijkstraComACM1975,
  title = {Guarded Commands, Nondeterminacy and Formal Derivation of Programs},
  author = {Dijkstra, Edsger W.},
  year = 1975,
  month = aug,
  journal = {Communications of the ACM},
  volume = {18},
  number = {8},
  pages = {453--457},
  issn = {0001-0782, 1557-7317},
  doi = {10.1145/360933.360975},
  url = {https://dl.acm.org/doi/10.1145/360933.360975},
  urldate = {2021-12-18},
  langid = {english}
}

@article{FengOOPSLA2023,
  title = {Lower {{Bounds}} for {{Possibly Divergent Probabilistic Programs}}},
  author = {Feng, Shenghua and Chen, Mingshuai and Su, Han and Kaminski, Benjamin Lucien and Katoen, Joost-Pieter and Zhan, Naijun},
  year = 2023,
  month = apr,
  journal = {Proceedings of the ACM on Programming Languages},
  volume = {7},
  number = {OOPSLA1},
  pages = {696--726},
  issn = {2475-1421},
  doi = {10.1145/3586051},
  url = {https://dl.acm.org/doi/10.1145/3586051},
  urldate = {2023-09-22},
  langid = {english}
}

@inproceedings{FerrerFioritiPOPL2015,
  title = {Probabilistic {{Termination}}: {{Soundness}}, {{Completeness}}, and {{Compositionality}}},
  shorttitle = {Probabilistic {{Termination}}},
  booktitle = {Proceedings of the 42nd {{Annual ACM SIGPLAN-SIGACT Symposium}} on {{Principles}} of {{Programming Languages}} - {{POPL}} '15},
  author = {Ferrer Fioriti, Luis Mar{\'i}a and Hermanns, Holger},
  year = 2015,
  pages = {489--501},
  publisher = {ACM Press},
  address = {Mumbai, India},
  doi = {10.1145/2676726.2677001},
  url = {http://dl.acm.org/citation.cfm?doid=2676726.2677001},
  urldate = {2020-03-02},
  isbn = {978-1-4503-3300-9},
  langid = {english}
}

@inproceedings{FuVMCAI2019,
  title = {Termination of {{Nondeterministic Probabilistic Programs}}},
  booktitle = {Verification, {{Model Checking}}, and {{Abstract Interpretation}}},
  author = {Fu, Hongfei and Chatterjee, Krishnendu},
  year = 2019,
  series = {Lecture {{Notes}} in {{Computer Science}}},
  volume = {11388},
  pages = {468--490},
  publisher = {Springer International Publishing},
  address = {Cham},
  doi = {10.1007/978-3-030-11245-5_22},
  url = {http://link.springer.com/10.1007/978-3-030-11245-5_22},
  urldate = {2020-03-15},
  isbn = {978-3-030-11244-8 978-3-030-11245-5},
  langid = {english}
}

@article{HarkPOPL2020,
  title = {Aiming Low Is Harder: Induction for Lower Bounds in Probabilistic Program Verification},
  shorttitle = {Aiming Low Is Harder},
  author = {Hark, Marcel and Kaminski, Benjamin Lucien and Giesl, J{\"u}rgen and Katoen, Joost-Pieter},
  year = 2020,
  month = jan,
  journal = {Proceedings of the ACM on Programming Languages},
  volume = {4},
  number = {POPL},
  pages = {1--28},
  issn = {2475-1421},
  doi = {10.1145/3371105},
  url = {https://dl.acm.org/doi/10.1145/3371105},
  urldate = {2021-05-17},
  langid = {english}
}

@inproceedings{HenzingerCAV2025,
  title = {Supermartingale Certificates for Quantitative Omega-Regular Verification and Control},
  booktitle = {Computer Aided Verification},
  author = {Henzinger, Thomas A. and Mallik, Kaushik and Sadeghi, Pouya and {\v Z}ikeli{\'c}, {\DJ}or{\dj}e},
  year = 2025,
  pages = {29--55},
  publisher = {Springer Nature Switzerland},
  address = {Cham},
  doi = {10.1007/978-3-031-98679-6_2},
  isbn = {978-3-031-98679-6}
}

@article{KaminskiJACM2018,
  title = {Weakest Precondition Reasoning for Expected Runtimes of Randomized Algorithms},
  author = {Kaminski, Benjamin Lucien and Katoen, Joost-Pieter and Matheja, Christoph and Olmedo, Federico},
  year = 2018,
  month = aug,
  journal = {Journal of the ACM},
  volume = {65},
  number = {5},
  pages = {1--68},
  issn = {00045411},
  doi = {10.1145/3208102},
  url = {http://dl.acm.org/citation.cfm?doid=3274534.3208102},
  urldate = {2020-03-02},
  langid = {english}
}

@inproceedings{KaminskiLICS2017,
  title = {A Weakest Pre-Expectation Semantics for Mixed-Sign Expectations},
  booktitle = {2017 32nd {{Annual ACM}}/{{IEEE Symposium}} on {{Logic}} in {{Computer Science}} ({{LICS}})},
  author = {Kaminski, Benjamin Lucien and Katoen, Joost-Pieter},
  year = 2017,
  month = jun,
  pages = {1--12},
  publisher = {IEEE},
  address = {Reykjavik, Iceland},
  doi = {10.1109/LICS.2017.8005153},
  url = {http://ieeexplore.ieee.org/document/8005153/},
  urldate = {2026-02-22},
  isbn = {978-1-5090-3018-7}
}

@inproceedings{Kenyon-RobertsLICS2021,
  title = {Supermartingales, {{Ranking Functions}} and {{Probabilistic Lambda Calculus}}},
  booktitle = {2021 36th {{Annual ACM}}/{{IEEE Symposium}} on {{Logic}} in {{Computer Science}} ({{LICS}})},
  author = {{Kenyon-Roberts}, Andrew and Ong, C.-H. Luke},
  year = 2021,
  month = jun,
  pages = {1--13},
  publisher = {IEEE},
  address = {Rome, Italy},
  doi = {10.1109/LICS52264.2021.9470550},
  url = {https://ieeexplore.ieee.org/document/9470550/},
  urldate = {2023-09-19},
  isbn = {978-1-6654-4895-6}
}

@inproceedings{KuraTACAS2019,
  title = {Tail Probabilities for Randomized Program Runtimes via Martingales for Higher Moments},
  booktitle = {Tools and {{Algorithms}} for the {{Construction}} and {{Analysis}} of {{Systems}}},
  author = {Kura, Satoshi and Urabe, Natsuki and Hasuo, Ichiro},
  year = 2019,
  series = {Lecture {{Notes}} in {{Computer Science}}},
  volume = {11428},
  pages = {135--153},
  publisher = {Springer},
  address = {Prague, Czech Republic},
  doi = {10.1007/978-3-030-17465-1_8},
  url = {https://doi.org/10.1007/978-3-030-17465-1_8}
}

@article{MajumdarPOPL2025,
  title = {Sound and {{Complete Proof Rules}} for {{Probabilistic Termination}}},
  author = {Majumdar, Rupak and Sathiyanarayana, V.R.},
  year = 2025,
  month = jan,
  journal = {Proceedings of the ACM on Programming Languages},
  volume = {9},
  number = {POPL},
  pages = {1871--1902},
  issn = {2475-1421},
  doi = {10.1145/3704899},
  url = {https://dl.acm.org/doi/10.1145/3704899},
  urldate = {2025-01-15},
  langid = {english}
}

@book{McIver2005,
  title = {Abstraction, Refinement and Proof for Probabilistic Systems},
  author = {McIver, Annabelle and Morgan, Carroll},
  year = 2005,
  series = {Monographs in Computer Science},
  publisher = {Springer},
  address = {New York},
  isbn = {978-0-387-27006-7},
  langid = {english},
  lccn = {005.1}
}

@article{McIverPOPL2018,
  title = {A New Proof Rule for Almost-Sure Termination},
  author = {McIver, Annabelle and Morgan, Carroll and Kaminski, Benjamin Lucien and Katoen, Joost-Pieter},
  year = 2018,
  month = jan,
  journal = {Proceedings of the ACM on Programming Languages},
  volume = {2},
  number = {POPL},
  pages = {1--28},
  issn = {2475-1421, 2475-1421},
  doi = {10.1145/3158121},
  url = {http://dl.acm.org/doi/10.1145/3158121},
  urldate = {2020-03-14},
  langid = {english}
}

@article{OlmedoTOPLAS2018,
  title = {Conditioning in {{Probabilistic Programming}}},
  author = {Olmedo, Federico and Gretz, Friedrich and Jansen, Nils and Kaminski, Benjamin Lucien and Katoen, Joost-Pieter and Mciver, Annabelle},
  year = 2018,
  month = mar,
  journal = {ACM Transactions on Programming Languages and Systems},
  volume = {40},
  number = {1},
  pages = {1--50},
  issn = {0164-0925, 1558-4593},
  doi = {10.1145/3156018},
  url = {https://dl.acm.org/doi/10.1145/3156018},
  urldate = {2023-10-08},
  langid = {english}
}

@inproceedings{SzymczakSETSS2020,
  title = {Weakest {{Preexpectation Semantics}} for {{Bayesian Inference}}: {{Conditioning}}, {{Continuous Distributions}} and {{Divergence}}},
  shorttitle = {Weakest {{Preexpectation Semantics}} for {{Bayesian Inference}}},
  booktitle = {Engineering {{Trustworthy Software Systems}}},
  author = {Szymczak, Marcin and Katoen, Joost-Pieter},
  year = 2020,
  series = {Lecture {{Notes}} in {{Computer Science}}},
  volume = {12154},
  pages = {44--121},
  publisher = {Springer International Publishing},
  address = {Cham},
  doi = {10.1007/978-3-030-55089-9_3},
  url = {http://link.springer.com/10.1007/978-3-030-55089-9_3},
  urldate = {2025-01-28},
  isbn = {978-3-030-55088-2 978-3-030-55089-9},
  langid = {english}
}

@inproceedings{TakisakaATVA2018,
  title = {Ranking and {{Repulsing Supermartingales}} for {{Reachability}} in {{Probabilistic Programs}}},
  booktitle = {Automated {{Technology}} for {{Verification}} and {{Analysis}}},
  author = {Takisaka, Toru and Oyabu, Yuichiro and Urabe, Natsuki and Hasuo, Ichiro},
  year = 2018,
  series = {Lecture {{Notes}} in {{Computer Science}}},
  volume = {11138},
  pages = {476--493},
  publisher = {Springer International Publishing},
  address = {Cham},
  doi = {10.1007/978-3-030-01090-4_28},
  url = {http://link.springer.com/10.1007/978-3-030-01090-4_28},
  urldate = {2020-03-29},
  isbn = {978-3-030-01089-8 978-3-030-01090-4}
}

@inproceedings{TakisakaCAV2024,
  title = {Lexicographic {{Ranking Supermartingales}} with {{Lazy Lower Bounds}}},
  booktitle = {Computer {{Aided Verification}}},
  author = {Takisaka, Toru and Zhang, Libo and Wang, Changjiang and Liu, Jiamou},
  year = 2024,
  series = {Lecture {{Notes}} in {{Computer Science}}},
  volume = {14683},
  pages = {420--442},
  publisher = {Springer Nature Switzerland},
  address = {Cham},
  doi = {10.1007/978-3-031-65633-0_19},
  url = {https://link.springer.com/10.1007/978-3-031-65633-0_19},
  urldate = {2024-11-29},
  isbn = {978-3-031-65632-3 978-3-031-65633-0},
  langid = {english}
}

@inproceedings{UrabeLICS2017,
  title = {Categorical Liveness Checking by Corecursive Algebras},
  booktitle = {2017 32nd {{Annual ACM}}/{{IEEE Symposium}} on {{Logic}} in {{Computer Science}} ({{LICS}})},
  author = {Urabe, Natsuki and Hara, Masaki and Hasuo, Ichiro},
  year = 2017,
  month = jun,
  pages = {1--12},
  publisher = {IEEE},
  address = {Reykjavik, Iceland},
  doi = {10.1109/LICS.2017.8005151},
  url = {http://ieeexplore.ieee.org/document/8005151/},
  urldate = {2020-03-27},
  isbn = {978-1-5090-3018-7}
}

@inproceedings{WangPLDI2019,
  title = {Cost Analysis of Nondeterministic Probabilistic Programs},
  booktitle = {Proceedings of the 40th {{ACM SIGPLAN Conference}} on {{Programming Language Design}} and {{Implementation}}},
  author = {Wang, Peixin and Fu, Hongfei and Goharshady, Amir Kafshdar and Chatterjee, Krishnendu and Qin, Xudong and Shi, Wenjun},
  year = 2019,
  month = jun,
  pages = {204--220},
  publisher = {ACM},
  address = {Phoenix AZ USA},
  doi = {10.1145/3314221.3314581},
  url = {https://dl.acm.org/doi/10.1145/3314221.3314581},
  urldate = {2025-10-23},
  isbn = {978-1-4503-6712-7},
  langid = {english}
}

@inproceedings{WangPLDI2021,
  title = {Quantitative Analysis of Assertion Violations in Probabilistic Programs},
  booktitle = {Proceedings of the 42nd {{ACM SIGPLAN International Conference}} on {{Programming Language Design}} and {{Implementation}}},
  author = {Wang, Jinyi and Sun, Yican and Fu, Hongfei and Chatterjee, Krishnendu and Goharshady, Amir Kafshdar},
  year = 2021,
  month = jun,
  pages = {1171--1186},
  publisher = {ACM},
  address = {Virtual Canada},
  doi = {10.1145/3453483.3454102},
  url = {https://dl.acm.org/doi/10.1145/3453483.3454102},
  urldate = {2025-04-18},
  isbn = {978-1-4503-8391-2},
  langid = {english}
}

@inproceedings{WangPLDI2021a,
  title = {Central Moment Analysis for Cost Accumulators in Probabilistic Programs},
  booktitle = {Proceedings of the 42nd {{ACM SIGPLAN International Conference}} on {{Programming Language Design}} and {{Implementation}}},
  author = {Wang, Di and Hoffmann, Jan and Reps, Thomas},
  year = 2021,
  month = jun,
  pages = {559--573},
  publisher = {ACM},
  address = {Virtual Canada},
  doi = {10.1145/3453483.3454062},
  url = {https://dl.acm.org/doi/10.1145/3453483.3454062},
  urldate = {2025-10-23},
  isbn = {978-1-4503-8391-2},
  langid = {english}
}

@article{WangPLDI2024,
  title = {Static {{Posterior Inference}} of {{Bayesian Probabilistic Programming}} via {{Polynomial Solving}}},
  author = {Wang, Peixin and Yang, Tengshun and Fu, Hongfei and Li, Guanyan and Ong, C.-H. Luke},
  year = 2024,
  month = jun,
  journal = {Proceedings of the ACM on Programming Languages},
  volume = {8},
  number = {PLDI},
  pages = {1361--1386},
  issn = {2475-1421},
  doi = {10.1145/3656432},
  url = {https://dl.acm.org/doi/10.1145/3656432},
  urldate = {2026-03-19},
  langid = {english}
}

@InProceedings{Kozen1979,
  author    = {Kozen, Dexter},
  booktitle = {20th Annual Symposium on Foundations of Computer Science (sfcs 1979)},
  title     = {Semantics of probabilistic programs},
  year      = {1979},
  month     = oct,
  pages     = {101--114},
  publisher = {IEEE},
  doi       = {10.1109/sfcs.1979.38},
}

@INPROCEEDINGS{Cook2005,
  author = {Byron Cook and Andreas Podelski and Andrey Rybalchenko},
  title = {Abstraction Refinement for Termination},
  booktitle = SAS2005,
  year = {2005},
  volume = {3672},
  series = LNCS,
  pages = {87--101},
  publisher = Springer,
  _isbn = {3-540-28584-9},
  file = {Cook2005.pdf:Cook2005.pdf:PDF},
  owner = {uhiro},
  timestamp = {2013.06.23}
}

@INPROCEEDINGS{Cook2006,
  author = {Byron Cook and Andreas Podelski and Andrey Rybalchenko},
  title = {Termination proofs for systems code},
  booktitle = PLDI2006,
  year = {2006},
  pages = {415--426},
  publisher = ACM,
  _isbn = {1-59593-320-4},
  file = {Cook2006.pdf:Cook2006.pdf:PDF},
  owner = {uhiro},
  timestamp = {2013.06.23}
}

@Article{Kura2024,
  author     = {Kura, Satoshi and Unno, Hiroshi},
  journal    = PACMPL,
  title      = {Automated Verification of Higher-Order Probabilistic Programs via a Dependent Refinement Type System},
  year       = {2024},
  month      = aug,
  number     = {ICFP},
  volume     = {8},
  _doi       = {10.1145/3674662},
  _url       = {https://doi.org/10.1145/3674662},
  address    = {New York, NY, USA},
  articleno  = {269},
  issue_date = {August 2024},
  numpages   = {30},
  publisher  = ACM,
}

@Article{Unno2023,
  author     = {Unno, Hiroshi and Terauchi, Tachio and Gu, Yu and Koskinen, Eric},
  journal    = PACMPL,
  title      = {Modular Primal-Dual Fixpoint Logic Solving for Temporal Verification},
  year       = {2023},
  month      = jan,
  number     = {POPL},
  volume     = {7},
  _address   = {New York, NY, USA},
  _doi       = {10.1145/3571265},
  _url       = {https://doi.org/10.1145/3571265},
  articleno  = {72},
  issue_date = {January 2023},
  numpages   = {30},
  publisher  = ACM,
}

\iflong
\clearpage
\appendix
\section{Formal Definition of the Weakest Preexpectation Transformer and its Representation as Fixed Point Equations}\label{appx:wp-formally}
Table~\ref{tab:wp} defines the weakest preexpectation transformer, and
Table~\ref{tab:equation-systems-wp} defines its translation to the fixed point equation system.

\begin{table}[b]
	\caption{The weakest preexpectation.
	For any $b : D \to \{ \mathbf{true}, \mathbf{false} \}$ and $f_1, f_2 : D \to [0, \infty]$, we define $\ifexpr{b}{f_1}{f_2} : D \to [0, \infty]$ by $(\ifexpr{b}{f_1}{f_2})(x) = f_1(x)$ if $b(x) = \mathbf{true}$ and $(\ifexpr{b}{f_1}{f_2})(x) = f_2(x)$ if $b(x) = \mathbf{false}$.}
	\label{tab:wp}
	\centering
	\begin{tabular}{l@{\hskip 1em}l}
		program $c$ \hspace{8em} & $\mathrm{wp}[c](f)$ \\
		\hline
		$\skipstmt$ & $f$ \\
		$c_1; c_2$ & $\mathrm{wp}[c_1](\mathrm{wp}[c_2](f))$ \\
		$\assign{x}{e}$ & $f[e/x]$ \\
		$\pbranch{c_1}{p}{c_2}$ & $\interpret{p} \cdot \mathrm{wp}[c_1](f) + (1 - \interpret{p}) \cdot \mathrm{wp}[c_2](f)$ \\
		$\ifstmt{\varphi}{c_1}{c_2}$ & $\ifexpr{\interpret{\varphi}}{\mathrm{wp}[c_1](f)}{\mathrm{wp}[c_2](f)}$ \\
		$\while{\varphi}{c}$ & $\mu \Phi^{\mathrm{wp}}$ \quad where $\Phi(X) \coloneqq \ifexpr{\interpret{\varphi}}{\mathrm{wp}[c](X)}{f}$
	\end{tabular}
\end{table}

\begin{table}[tbp]
	\caption{The translation into equation systems for weakest preexpectation.
		Here $(F_i, E_i) = \mathrm{wp}'[c_i](F, E)$ for $i = 1, 2$ and $(F', E') = \mathrm{wp}'[c](X(\widetilde{x}), E)$.
		$F[e/x]$ denotes substitution.}
	\label{tab:equation-systems-wp}
	\centering
	\begin{tabular}{l@{\hskip 1em}l}
		program $c$ \hspace{8em} & $\mathrm{wp}'[c](F, E)$ \\
		\hline
		$\skipstmt$ & $(F, E)$ \\
		$c_1; c_2$ & $\mathrm{wp}'[c_1](\mathrm{wp}'[c_2](F, E))$ \\
		$\assign{x}{e}$ & $(F[e/x], E)$ \\
		$\pbranch{c_1}{p}{c_2}$ & $(p \cdot F_1 + (1 - p) \cdot F_2, E_1 \cup E_2)$ \\
		$\ifstmt{\varphi}{c_1}{c_2}$ & $(\ifexpr{\varphi}{F_1}{F_2}, E_1 \cup E_2)$ \\
		$\while{\varphi}{c}$ & $(X(\widetilde{x}), E' \cup \{ X(\widetilde{x}) =_{\mu} \ifexpr{\varphi}{F'}{F} \})$
	\end{tabular}
\end{table}

\section{Proof of Lemma~\ref{lem:affine-k-induces-linear-knt}}
	We prove \( \nonterm{K}(\alpha \eta) = \alpha \nonterm{K}(\eta) \).
	\begin{itemize}
		\item Case \( \alpha = 0 \): Then \( \nonterm{K}(0 \eta) = \nonterm{K}(\probzero) = K(\probzero) - K(\probzero) = 0 = \alpha \nonterm{K}(\eta) \).
		\item Case \( \alpha = 1 \): Trivial.
		\item Case \( \alpha \in (0,1) \): Since \( K \) is affine,
			\begin{equation*}
				K(\alpha \eta) = K(\alpha \eta + (1-\alpha)\probzero) = \alpha K(\eta) + (1-\alpha)K(\probzero).
			\end{equation*}
			If \( K(\alpha \eta) < \infty \), then \( K(\eta), K(\probzero) < \infty \), and
			\begin{align*}
				\nonterm{K}(\alpha \eta)
				&= K(\alpha \eta) - K(\probzero)
				\\
				&= \alpha K(\eta) + (1-\alpha)K(\probzero) - K(\probzero)
				\\
				&= \alpha K(\eta) - \alpha K(\probzero)
				\\
				&= \alpha \nonterm{K}(\eta).
			\end{align*}
			Assume \( K(\alpha \eta) = \infty \).
			Then \( K(\eta) = \infty \).
			If \( K(\probzero) \neq \infty \), then both \( \alpha \nonterm{K}(\eta) = \alpha K(\eta) - \alpha K(\probzero) \) and \( \nonterm{K}(\alpha \eta) = K(\alpha \eta) - K(\probzero) \) are \( \infty \).
			If \( K(\probzero) = \infty \), by the monotonicity of \( K \), it is the constant function to \( \infty \).
			So \( \nonterm{K} \) is the constant function to \( 0 \), which trivially satisfies the condition.
		\item Case \( \alpha \in (1, \infty) \):
			Applying the above argument to \( \beta := (1/\alpha) \), we have \( \nonterm{K}(\beta \eta') = \beta \nonterm{K}(\eta') \) for every \( \eta' \).
			We obtain the claim by taking \( \eta' := \alpha \eta \) and multiplying \( \alpha \) to both sides.
	\end{itemize}

	We prove \( \nonterm{K}(\eta_1 + \eta_2) = \nonterm{K}(\eta_1 + \eta_2) \):
	\begin{align*}
		\nonterm{K}(\eta_1 + \eta_2)
		&= K(\eta_1 + \eta_2) - K(\probzero)
		\\
		&= K\left(\dfrac{1}{2} (2\eta_1) + \dfrac{1}{2}(2\eta_2)\right) - K(\probzero)
		\\
		&= \dfrac{1}{2}K(2\eta_1) + \dfrac{1}{2}K(2\eta_2)-K(0)
		\\
		&= \dfrac{1}{2}\nonterm{K}(2\eta_1) + \dfrac{1}{2}\nonterm{K}(2\eta_2)
		\\
		&= \nonterm{K}(\eta_1) + \nonterm{K}(\eta_2).
	\end{align*}

\iflong
\else
\section{Remarks on Greatest Fixed Points}\label{sec:upper-bound-gfp}
There are also quantitative properties of probabilistic programs that can be characterised by the greatest fixed point (e.g.\ the weakest liberal preexpectation and the conditional weakest preexpectation~\cite{OlmedoTOPLAS2018}).
We comment on the applicability of our result to greatest fixed points.

It is straightforward to apply our result to 1-bounded equation systems.
Once we establish the uniqueness of fixed points by Theorem~\ref{thm:unique-fixed-point-1-bounded}, we can give an upper bound of the greatest fixed point by a prefixed point.
Alternatively, we can consider the (order-reversing) isomorphism $1 - ({-}) : ([0, 1], {\le}) \to ([0, 1], {\ge})$ to translate greatest fixed points to least fixed points.

For unbounded expectations, however, it is not clear how we can use our result because the restriction to $\mathbb{E}_{\le u}(D)$ changes the greatest fixed point: $\nu \interpret{E} = \inf_n \interpret{E}^n(\mathbf{\infty})$ may be different from $\nu \restrictby{\interpret{E}}{u} = \inf_n \interpret{E}^n(u)$.
In practice, most of the properties characterised by the greatest fixed point are 1-bounded, and thus, we do not consider this problem in this paper.

Mixing least and greatest fixed points is often studied for verification of non-probabilistic programs.
For example, model checking for modal $\mu$-calculus is an example of such a problem.
Mixing least and greatest fixed points in our setting is an interesting direction, but we leave it for future work.
\fi

\section{More Examples}

\begin{example}\label{ex:olmedo18}
	The conditional weakest preexpectation is defined as follows.
	\[ \mathrm{cwp}[c](f) \quad\coloneqq\quad \frac{\mathrm{cwp}_1[c](f)}{\probone - \mathrm{cwp}_2[c](\probzero)} \]
	We can give a translation to equation systems as in Section~\ref{sec:soft-hard-conditioning}.
	If we have lower bounds $l_1 \le \mathrm{cwp}_1[c](f)$ and $l_2 \le \mathrm{cwp}_2[c](\probzero)$, then a lower bound of $\mathrm{cwp}[c](f)$ is given as follows.
	\[ \mathrm{cwp}[c](f) \quad=\quad \frac{\mathrm{cwp}_1[c](f)}{1 - \mathrm{cwp}_2[c](\probzero)} \quad\ge\quad \frac{l_1}{1 - l_2} \]

	Now, consider the following program taken from~\cite[Example~4.4]{OlmedoTOPLAS2018}.
	\begin{align*}
		c_{\mathrm{tails}} \quad&=\quad \assign{m}{0};\ \assign{b_1, b_2, b_3}{\mathbf{true}};\\
		&\qquad\qquad\while{b_1 \lor b_2 \lor b_3}{\\
			&\qquad\qquad\qquad (\pbranch{\assign{b_1}{\mathbf{true}}}{1/2}{\assign{b_1}{\mathbf{false}}});\\
			&\qquad\qquad\qquad (\pbranch{\assign{b_2}{\mathbf{true}}}{1/2}{\assign{b_2}{\mathbf{false}}});\\
			&\qquad\qquad\qquad (\pbranch{\assign{b_3}{\mathbf{true}}}{1/2}{\assign{b_3}{\mathbf{false}}});\\
			&\qquad\qquad\qquad \observe{\lnot b_1 \lor \lnot b_2 \lor \lnot b_3};\\
			&\qquad\qquad\qquad \assign{m}{m + 1}\\
		&\qquad\qquad}
	\end{align*}
	Here, $m$ is an integer variable, $b_1, b_2, b_3$ are boolean variables.
	We aim to obtain a lower bound of $\mathrm{cwp}[c_{\mathrm{tails}}]([m = N])$ where $N > 0$.
	We first consider the equation system for $\mathrm{cwp}_1[c_{\mathrm{tail}}]([m = N])$, which is given as follows.
	\[ X_1(m, b_1, b_2, b_3) \quad=_{\mu}\quad \ifexpr{b_1 \lor b_2 \lor b_3}{F_1}{[m = N]} \]
	where $\mathrm{cwp}_1[c_{\mathrm{tail}}]([m = N]) = X_1(0, \mathbf{true}, \mathbf{true}, \mathbf{true})$ and $F_1$ is defined as follows.
	\[ F_1 \quad=\quad \sum_{b_1, b_2, b_3 \in \{ \mathbf{true}, \mathbf{false} \}} \frac{1}{8} \cdot \ifexpr{\lnot b_1 \lor \lnot b_2 \lor \lnot b_3}{X_1(m+1, b_1, b_2, b_3)}{0}. \]
	A lower-bound certificate for $X_1(m, b_1, b_2, b_3)$ is given as follows.
	\begin{align*}
		&\text{prefixed point} & u_1(m, b_1, b_2, b_3) &= 1 \\
		&\text{$u$-ranking supermartingale} & r_1(m, b_1, b_2, b_3) &= 8 \\
		&\text{postfixed point} & l_1(m, b_1, b_2, b_3) &= \begin{cases}
			[m = N] & b_1 \lor b_2 \lor b_3 = \mathbf{false} \\
			\frac{1}{6} \left( \frac{3}{4} \right)^{N - m} & (b_1 \lor b_2 \lor b_3 = \mathbf{true}) \land m < N \\
			0 & (b_1 \lor b_2 \lor b_3 = \mathbf{true}) \land m \ge N
		\end{cases}
	\end{align*}

	The equation system for $\mathrm{cwp}_2[c_{\mathrm{tail}}](\probzero)$ is given as follows.
	\[ X_2(m, b_1, b_2, b_3) \quad=_{\mu}\quad \ifexpr{b_1 \lor b_2 \lor b_3}{F_2}{0} \]
	where $\mathrm{cwp}_2[c_{\mathrm{tail}}](\probzero) = X_2(0, \mathbf{true}, \mathbf{true}, \mathbf{true})$ and
	\[ F_2 \quad=\quad \sum_{b_1, b_2, b_3 \in \{ \mathbf{true}, \mathbf{false} \}} \frac{1}{8} \cdot \ifexpr{\lnot b_1 \lor \lnot b_2 \lor \lnot b_3}{X_2(m+1, b_1, b_2, b_3)}{1}. \]
	A lower-bound certificate for $X_2(m, b_1, b_2, b_3)$ is given as follows.
	\begin{align*}
		&\text{prefixed point} & u_2(m, b_1, b_2, b_3) &= 1 \\
		&\text{$u$-ranking supermartingale} & r_2(m, b_1, b_2, b_3) &= 8 \\
		&\text{postfixed point} & l_2(m, b_1, b_2, b_3) &= \begin{cases}
			0 & b_1 \lor b_2 \lor b_3 = \mathbf{false} \\
			\frac{1}{2} & b_1 \lor b_2 \lor b_3 = \mathbf{true}
		\end{cases}
	\end{align*}
	Therefore, we obtain the following lower bound, which is actually the exact value of $\mathrm{cwp}[c_{\mathrm{tails}}]([m = N])$.
	\[ \mathrm{cwp}[c_{\mathrm{tails}}]([m = N]) \quad\ge\quad \frac{l_1(0, \mathbf{true}, \mathbf{true}, \mathbf{true})}{1 - l_2(0, \mathbf{true}, \mathbf{true}, \mathbf{true})} = \frac{1}{3} \left( \frac{3}{4} \right)^{N} \]
\end{example}

\begin{example}[expected runtime of the coupon collector]\label{ex:coupon-collector-exact}
	Consider the following equation system for the coupon collector problem \cite[Example~47]{HarkPOPL2020}.
	\begin{align*}
		X(x) \quad&=_{\mu}\quad \ifexpr{0 < x}{Y(x, N + 1)}{0} \\
		Y(x, i) \quad&=_{\mu}\quad \ifexpr{x < i}{1 + \frac{1}{N} \sum_{i' = 1}^N Y(x, i')}{X(x - 1)}
	\end{align*}
	Here, $x$ and $i$ are integers.
	The least fixed point is given as follows.
	\begin{align*}
		X(x) &= \begin{cases}
			0 & x \le 0 \\
			N H_x & 0 < x \le N \\
			N H_N & N < x
		\end{cases} &
		Y(x, i) &= \begin{cases}
			\infty & x < i \land x \le 0 \\
			N H_{x} & x < i \land 0 < x \le N \\
			1 + N H_{N} & x < i \land N < x \\
			0 & x \ge i \land x \le 0 \\
			N H_{x - 1} & x \ge i \land 0 < x \le N \\
			N H_{N} & x \ge i \land N < x
		\end{cases}
	\end{align*}
	where $H_n = \sum_{k = 1}^n 1 / k$ is the $n$-th harmonic number.
	We consider the following under-approximate equation system.
	\begin{align*}
		X'(x) \quad&=_{\mu}\quad \ifexpr{0 < x}{Y'(x, N + 1)}{0} \\
		Y'(x, i) \quad&=_{\mu}\quad \ifexpr{x \le 0}{0}{\ifexpr{x < i}{1 + \frac{1}{N} \sum_{i' = 1}^N Y(x, i')}{X(x - 1)}}
	\end{align*}
	Then, we have the following pre- and post-fixed point.
	\begin{align*}
		u_X(x) = l_X(x) &= \begin{cases}
			0 & x \le 0 \\
			N H_x & 0 < x \le N \\
			N H_N & N < x
		\end{cases} \\
		u_Y(x, i) = l_Y(x, i) &= \begin{cases}
			0 & x \le 0 \\
			N H_{x} & x < i \land 0 < x \le N \\
			1 + N H_{N} & x < i \land N < x \\
			N H_{x - 1} & x \ge i \land 0 < x \le N \\
			N H_{N} & x \ge i \land N < x
		\end{cases}
	\end{align*}
	As for the $u$-ranking supermartingale, it suffices to give a $1$-ranking supermartingale because $u_X$ and $u_Y$ are bounded by $1 + N H_N$.
	A $1$-ranking supermartingale is given as follows.
	\begin{align*}
		r_X(x) &= \begin{cases}
			1 & x \le 0 \\
			2 N H_x + 2 x + 1 & 0 < x \le N \\
			2 N H_N + 2 x + 1 & N < x
		\end{cases} \qquad&
		r_Y(x, i) &= \begin{cases}
			1 & x \le 0 \\
			2 N H_{x} + 2 x & x < i \land 0 < x \le N \\
			1 + N H_{N} + 2 x & x < i \land N < x \\
			N H_{x - 1} + 2 x & x \ge i \land 0 < x \le N \\
			N H_{N} + 2 x & x \ge i \land N < x
		\end{cases}
	\end{align*}
	Then, the pair of $(1 + N H_N) \cdot r_X$ and $(1 + N H_N) \cdot r_Y$ is a $u$-ranking supermartingale.
\end{example}

\begin{example}
	We give an approximate lower bound of Example~\ref{ex:coupon-collector-exact} using logarithmic bounds.
	Recall that the $n$-th harmonic number $H_n$ has the following logarithmic bound.
	\[ \ln (n + 1) \quad\le\quad H_n \quad\le\quad \ln n + 1 \qquad (n \ge 1) \]
	We can verify that the following function is a prefixed point of the equation system.
	\begin{align*}
		u_X(x) &= \begin{cases}
			0 & x \le 0 \\
			N (\ln x + 1) & 0 < x \le N \\
			N (\ln N + 1) & N < x
		\end{cases} \qquad&
		u_Y(x, i) &= \begin{cases}
			0 & x \le 0 \\
			N (\ln x + 1) & x < i \land 0 < x \le N \\
			1 + (\ln N + 1) & x < i \land N < x \\
			0 & x \ge i \land x = 1 \\
			N (\ln (x - 1) + 1) & x \ge i \land 1 < x \le N \\
			N (\ln N + 1) & x \ge i \land N < x
		\end{cases}
	\end{align*}
	The non-trivial part is the case $x < i \land 0 < x \le N$ for $u_Y$.
	If $x > 1$, then we have the following.
	\begin{align*}
		1 + \frac{1}{N} \sum_{i' = 1}^N u_Y(x, i') &= 1 + \frac{x}{N} \cdot N (\ln (x - 1) + 1) + \frac{N - x}{N} \cdot N (\ln x + 1) \\
		&= 1 + x \ln \frac{x - 1}{x} + N (\ln x + 1) \\
		&\le N (\ln x + 1) \\
		&= u_Y(x, i)
	\end{align*}
	If $x = 1$, then we have the following.
	\[ 1 + \frac{1}{N} \sum_{i' = 1}^N u_Y(x, i') = 1 + \frac{N - 1}{N} \cdot N = N = u_Y(x, i) \]
	To give a $u$-ranking supermartingale, it suffices to give a $1$-ranking supermartingale because $u_X$ and $u_Y$ are bounded.
	\begin{align*}
		r_X(x) &= \begin{cases}
			1 & x \le 0 \\
			2 N (\ln x + 1) + 2 x + 1 & 0 < x \le N \\
			2 N (\ln N + 1) + 2 x + 1 & N < x
		\end{cases} \\
		r_Y(x, i) &= \begin{cases}
			1 & x \le 0 \\
			2 N (\ln x + 1) + 2 x & x < i \land 0 < x \le N \\
			1 + N (\ln N + 1) + 2 x & x < i \land N < x \\
			2 & x \ge i \land x = 1 \\
			N (\ln (x - 1) + 1) + 2 x & x \ge i \land 1 < x \le N \\
			N (\ln N + 1) + 2 x & x \ge i \land N < x
		\end{cases}
	\end{align*}
	A postfixed point is given as follows.
	\begin{align*}
		l_X(x) &= \begin{cases}
			0 & x \le 0 \\
			N \ln (x + 1) & 0 < x \le N \\
			N \ln (N + 1) & N < x
		\end{cases} \qquad&
		l_Y(x, i) &= \begin{cases}
			0 & x \le 0 \\
			N \ln (x + 1) & x < i \land 0 < x \le N \\
			1 + N \ln (N + 1) & x < i \land N < x \\
			N \ln x & x \ge i \land 0 < x \le N \\
			N \ln (N + 1) & x \ge i \land N < x
		\end{cases}
	\end{align*}
	Again, the non-trivial part is the case $x < i \land 0 < x \le N$ for $u_Y$.
	If $x > 1$, then we have the following.
	\begin{align*}
		1 + \frac{1}{N} \sum_{i' = 1}^N l_Y(x, i') &= 1 + \frac{x}{N} \cdot N \ln x + \frac{N - x}{N} \cdot N \ln (x + 1) \\
		&= 1 + x \ln \frac{x}{x + 1} + N \ln (x + 1) \\
		&\ge N \ln (x + 1) \\
		&= l_Y(x, i)
	\end{align*}
	If $x = 1$, then we have the following.
	\[ 1 + \frac{1}{N} \sum_{i' = 1}^N l_Y(x, i') = 1 + \frac{N - 1}{N} \cdot N \ln 2 = N \ln 2 + 1 - \ln 2 \ge l_Y(x, i) \]
\end{example}

\section{Simulating Existing Techniques in Our Framework}
\label{sec:simulating-existing-techniques-by-underapproximation}

\subsection{Simulating $\gamma$-Scaled Submartingales}
We can simulate the idea of $\gamma$-scaled submartingales by scaling the right-hand sides of a 1-bounded equation system by $0 < \gamma < 1$.
Given an equation system
\[ E \quad=\quad \{\ X_1 =_{\mu} F_1, \dots, X_n =_{\mu} F_n\ \} \]
we define the $\gamma$-scaled equation system as follows.
\[ E' \quad=\quad \{\ X_1 =_{\mu} \gamma \cdot F_1, \dots, X_n =_{\mu} \gamma \cdot F_n\ \} \]
Here, the scalar multiplication $\gamma \cdot ({-})$ is extended to quantitative formulas in the natural way.

Then, $E'$ has a unique fixed point and satisfies $\interpret{E'} \le \interpret{E}$.
Since $E$ is 1-bounded, we have $\interpret{E}(\mathbf{1}) \le \mathbf{1}$ and $\interpret{\nonterm{E}}(\mathbf{1}) \le \mathbf{1}$.
By definition of $E'$, we have $\interpret{E'} = \gamma \cdot \interpret{E}$ and $\interpret{\nonterm{E'}} = \gamma \cdot \interpret{\nonterm{E}}$.
Therefore, the constant function $r(x) = 1 / (1 - \gamma)$ is a 1-ranking supermartingale for $E'$.
\[ \interpret{\nonterm{E'}}(r) + \mathbf{1} \quad=\quad \frac{\gamma}{1 - \gamma} \cdot \interpret{\nonterm{E}}(\mathbf{1}) + \mathbf{1} \quad\le\quad \frac{1}{1 - \gamma} \cdot \mathbf{1} \quad=\quad r \]

\subsection{Simulating Guard-Strengthening}\label{sec:guard-strengthening}
The guard-strengthening~\cite[Corollary 13]{FengOOPSLA2023} is a way to obtain $E'$ from $E$ such that $\interpret{E'} \le \interpret{E}$ holds and $E'$ has a unique fixed point.
Given a program of the form $\mathbf{while}\ (\varphi) \ \{ C \}$, the guard-strengthening~\cite[Corollary 13]{FengOOPSLA2023} gives a lower bound of $\mathrm{wp}[\mathbf{while}\ (\varphi) \ \{ C \}](f)$ as $\mathrm{wp}[\mathbf{while}\ (\varphi') \ \{ C \}]([\lnot \varphi] \cdot f)$.

We can simulate the guard-strengthening in our framework.
The weakest preexpectation $\mathrm{wp}[\while{\varphi}{C}](f)$ is translated into the following equation system.
\[ E \quad=\quad \{ X(x) =_{\mu} \ifexpr{\varphi}{\mathrm{wp}[C](X)}{f} \} \]
The guard-strengthening is simulated by defining $E'$ as follows.
\begin{equation}
	E' \quad=\quad \{ X'(x) =_{\mu} \ifexpr{\varphi' \lor \lnot \varphi}{(\ifexpr{\varphi}{\mathrm{wp}[C](X')}{f})}{0} \}
	\label{eq:guard-strengthening}
\end{equation}
The quantitative predicate $X'$ is defined in the same way as $X$ except when neither $\varphi'$ nor $\lnot \varphi$ holds, in which case the value of $X'$ is truncated to $0$.
That is, we have the following equation.
\[ \interpret{E'}(X')(x) = \begin{cases}
	\interpret{E}(X')(x) & \text{if $x$ satisfies $\varphi' \lor \lnot \varphi$} \\
	0 & \text{otherwise}
\end{cases} \]
Therefore, we have $\interpret{E'} \le \interpret{E}$.
If $\varphi' \implies \varphi$ holds, then we can rewrite \eqref{eq:guard-strengthening} as follows.
\[ X'(x) =_{\mu} \ifexpr{\varphi'}{\mathrm{wp}[C](X')}{[\lnot \varphi] \cdot f} \]
This coincides with the equation system for $\mathrm{wp}[\while{\varphi'}{C}]([\lnot \varphi] \cdot f)$.

\subsection{Simulating Stochastic Invariants}
\label{sec:simulating-stochastic-invariants}
Stochastic invariants provide a method~\cite{ChatterjeeCAV2022} to reason about lower bounds of termination probabilities.
A stochastic invariant $(I, p)$ is a pair of a set $I$ of states of a probabilistic system and $p \in [0, 1]$ such that the probability of the system leaving $I$ is at most $p$.
If there exists a stochastic invariant $(I, p)$ and the probabilistic system almost surely terminates or leaves $I$, then the termination probability is at least $1 - p$.
We can simulate the idea of stochastic invariants in our framework.
The basic idea is similar to guard-strengthening.
We define an equation system $E'$ such that $\interpret{E'} \le \interpret{E}$ by restricting the domain to the set $I$.
Then, a termination certificate ensures that $\interpret{E'}$ has a unique fixed point.
A non-trivial part is that a certificate for the stochastic invariant, which is a prefixed point of some function, can be turned into a postfixed point by inverting a certificate by $1 - ({-})$.
We will describe this in detail below.

For simplicity, suppose that we have a while loop $\while{\varphi}{c}$.
Then, the termination probability is given by $\mathrm{wp}[\while{\varphi}{c}](1)$, which is the solution of the following equation system.
\[ E \quad=\quad \{ X =_{\mu} \ifexpr{\varphi}{\mathrm{wp}[c](X)}{1} \} \]
Let $D$ be the set of program states.
Observe that $\interpret{E} : \mathbb{E}(D) \to \mathbb{E}(D)$ is an affine function such that $\interpret{E} (1) = 1$.
Note that we have $\interpret{E} (1) = 1$ for the case of termination probability problems because by Table~\ref{tab:equation-systems-wp}, if $F$ and $E$ have total weight 1, then so is $\mathrm{wp}'[c](F, E)$.
Here, we say $K(\mathbf{1}) \in \mathbb{E}(D)$ is the \emph{total weight} of $K : \mathbb{E}(D) \to \mathbb{E}(D)$.

Now, suppose we have a stochastic invariant $(I, p)$ such that the probabilistic system almost surely terminates or leaves $I$.
We assume that the stochastic invariant is witnessed by a stochastic invariant indicator $f : D \to [0, \infty)$~\cite{ChatterjeeCAV2022} and that the almost sure termination is witnessed by a ranking supermartingale $r : D \to [0, \infty)$.
Concretely, $f$ and $r$ satisfy the following conditions.
\begin{itemize}
	\item $\interpret{E} (f) (x) \le f(x)$ for any $x \in I$
	\item $f(x) \ge 1$ for any $x \notin I$
	\item $f(x_0) \le p$ where $x_0$ is the initial state
	\item $\nonterm{\interpret{E}} (r) (x) + 1 \le r(x)$ for any $x \in I$
\end{itemize}
In this situation, we can define a lower-bound certificate for the termination probability as follows.
\begin{itemize}
	\item We define $E'$ such that $\interpret{E'} \le \interpret{E}$ as follows.
	\[ E' \quad=\quad \{ X =_{\mu} F'(X) \} \quad\text{where}\quad F' (X) (x) = \chi_I \cdot \interpret{E}(X)(x) \]
	This definition is the same that for the guard-strengthening described in Section~\ref{sec:guard-strengthening}.
	Note that $\nonterm{\interpret{E'}} = \nonterm{F'}$ is given as follows.
	\[ \nonterm{F'} (X) (x) = \chi_I \cdot \nonterm{\interpret{E}}(X)(x) \]
	\item We define $u = 1$. Then, $u$ is a prefixed point of $\interpret{E'}$ because $\interpret{E'} 1 \le \interpret{E} 1 = 1$ holds.
	\item We define $r' = r + 1$. Then, $r'$ is a $u$-ranking supermartingale.
	We have $\nonterm{\interpret{E'}} (r') (x) + 1 \le r'(x)$ for any $x \in I$.
	If $x \in I$, then
	\begin{align*}
		\nonterm{\interpret{E'}} (r') (x) + 1 &= \nonterm{\interpret{E}}(r')(x) + 1 \\
		&\le \nonterm{\interpret{E}} (r)(x) + 2 \\
		&\le r(x) + 1 \\
		&= r'(x).
	\end{align*}
	If $x \notin I$, then $\nonterm{\interpret{E'}} (r') (x) + 1 = 1 \le r'(x)$.
	Here, we have $r' \ge 1$ because we defined $r'$ as $r' = r + 1$.
	\item We define $f' = \min \{ f, 1 \}$. Then, we have $\interpret{E} (f') (x) \le f'(x)$ for any $x \in I$ and $f'(x) = 1$ for any $x \notin I$.
	This is a \emph{prefixed} point of $X \mapsto \ifexpr{({-}) \in I}{\interpret{E}(X)}{1}$.
	Now, we define $l = 1 - f'$.
	Then, $l$ is a \emph{postfixed} point of $\interpret{E'}$.
	If $x \in I$, then
	\[ 1 - f'(x) \le 1 - \interpret{E} (f')(x) \le \interpret{E}(1 - f')(x). \]
	If $x \notin I$, then
	\[ 1 - f'(x) = 0 = \interpret{E'} (1 - f')(x). \]
	Obviously, $l \le u$ also holds.
\end{itemize}
This lower-bound certificate witnesses that the termination probability is at least $l$.
Specifically, if the initial state is $x_0$, then the termination probability is at least $l(x_0) \ge 1 - p$.

\section{Nondeterminism}

\subsection{Demonic Nondeterminism}\label{sec:demonic-nondeterminism}

We consider extending our results to demonic nondeterminism, which means that we aim to obtain a lower bound for \emph{any} scheduler that resolves nondeterminism.
Specifically, we consider fixed-point equation systems that contain the $\min$ operator.
\begin{equation*}
	F \quad\coloneqq\quad X(\widetilde{e}) \mid t \mid F_1 + F_2 \mid t \cdot F \mid \ifexpr{\varphi}{F_1}{F_2} \mid \min \{ F_1, \dots, F_n \}
\end{equation*}
The interpretation of quantitative formulas is extended as follows.
\[ \interpret{\min \{ F_1, \dots, F_n \}} \quad\coloneqq\quad \min \{ \interpret{F_1}, \dots, \interpret{F_n} \} \]

\begin{example}[demonically fair random walk]\label{ex:demonically-fair-random-walk}
	\cite[Section~5.2]{McIverPOPL2018}
	We write $\square$ for nondeterministic branching.
	\[ \while{x > 0}{\pbranch{\assign{x}{x - 1}}{1 / 2}{(\assign{x}{x + 1} \mathrel{\square} \skipstmt)}} \]
	The weakest preexpectation is given using the $\min$ operator \cite{McIver2005}.
	\[ X(x) \quad=_{\mu}\quad \ifexpr{x > 0}{\frac{1}{2} X(x - 1) + \frac{1}{2} \min \{ X(x + 1), X(x) \}}{1} \tag*{\qed} \]
\end{example}
Quantitative formulas extended with the $\min$ operator are not in the scope of Corollary~\ref{cor:unique-fixed-point-by-ranking-supermartingale} because the $\min$ operator is not affine.
We need to extend our framework to handle the $\min$ operator.

For simplicity, consider the equation system of the following form:
\begin{equation}
	X \quad=_{\mu}\quad \min \{ F_1, \dots, F_n \} \label{eq:min-equation}
\end{equation}
where $F_1, \dots, F_n$ are quantitative formulas over $\{ X \}$ that do not contain the $\min$ operator.
Note that we can always rewrite the equation system into this form by distributing $\min$ over other operators.
\begin{lemma}\label{lem:min-distributive}
	The $\min$ operator is distributive over addition, scalar multiplication, and the if-then-else operator as follows.
	The same equations hold if we replace each $\min$ with $\max$.
	\[ \min \{ F_1, \dots, F_m \} + \min \{ G_1, \dots, G_n \} \quad=\quad \min \{ F_i + G_j \mid i = 1, \dots, m; j = 1, \dots, n \} \]
	\[ t \cdot \min \{ F_1, \dots, F_n \} \quad=\quad \min \{ t \cdot F_1, \dots, t \cdot F_n \} \]
	\begin{align*}
		&\ifexpr{\varphi}{\min \{ F_1, \dots, F_m \}}{\min \{ G_1, \dots, G_n \}} \\
		&=\quad \min \{ \ifexpr{\varphi}{F_i}{G_j} \mid i = 1, \dots, m; j = 1, \dots, n \}
	\end{align*}
\end{lemma}
\begin{proof}
	By straightforward calculations.
\end{proof}

Although the equation system~\eqref{eq:min-equation} itself is not affine, we can transform this into (a collection of) affine equation systems by the following proposition.

\begin{proposition}\label{prop:min-affinise}
	Let $K_1, \dots, K_n : \mathbb{E}(D) \to \mathbb{E}(D)$ be affine functions.
	Then, we have the following.
	\[ \mu X. \min \{ K_1(X), \dots, K_n(X) \} \qquad=\qquad \inf_{\substack{t_1, \dots, t_n : D \to [0, 1] \\ t_1 + \dots + t_n = \probone}} \mu X. \sum_{i = 1}^n t_i \cdot K_i(X) \]
\end{proposition}
\begin{proof}
	It is immediate to show $\le$ because $\min \{ K_1, \dots, K_n \} \le \sum_{i = 1}^n t_i \cdot K_i$ for any $t_1, \dots, t_n : D \to [0, 1]$ such that $t_1 + \dots + t_n = \probone$.
	For the converse direction, let $K \coloneqq \min \{ K_1, \dots, K_n \}$ and $t_i : D \to \{0, 1\}$ be the function such that $t_i(x) = 1$ if $i = \argmin_j K_j(\mu X. K(X))(x)$ and $t_i(x) = 0$ otherwise.
	Here, $\argmin$ returns the smallest index when there are multiple indices that minimize the value so that we have $t_1 + \dots + t_n = \probone$.
	Then, we have for any $x \in D$,
	\[ (\mu X. K(X))(x) \quad=\quad K(\mu X. K(X))(x) \quad=\quad \sum_{i = 1}^n t_i \cdot K_i(\mu X. K(X))(x). \]
	That is, $\mu X. K(X)$ is a fixed point of $\sum_{i = 1}^n t_i \cdot K_i$, so we have $\mu X. K(X) \ge \mu X. \sum_{i = 1}^n t_i \cdot K_i(X)$.
\end{proof}

Now, we extend the reasoning principle in Theorem~\ref{thm:unbounded-reasoning-principle} to the equation system~\eqref{eq:min-equation}.
By Proposition~\ref{prop:min-affinise}, we can give a lower bound if there exists a common lower-bound certificate for any $t_1, \dots, t_n : D \to [0, 1]$ such that $t_1 + \dots + t_n = \probone$.
Concretely:
\begin{itemize}
	\item A prefixed point $u$ such that
	$\sum_{i = 1}^n t_i \cdot \interpret{F_i}(u) \le u$.
	\item A $u$-ranking supermartingale $r$ such that
	$\sum_{i = 1}^n t_i \cdot \interpret{\nonterm{F_i}}(r) + u \le r$.
	\item An invariant $\eta$ such that $\eta \le u$ and
	$\sum_{i = 1}^n t_i \cdot \interpret{F_i}(\eta) \ge \eta$.
\end{itemize}

We can simplify this condition.
Actually, it is sufficient to consider only the case where $t_i = \probone$ for some $i$, and the above conditions are equivalent to the following.
\begin{theorem}
	Consider the equation system $X =_{\mu} F$.
	If we have the following data, then $\eta$ is a lower bound of the least fixed point of the equation system.
	\begin{itemize}
		\item A prefixed point $u$ of $F^{\max}$:
		\[ \interpret{F^{\max}}(u) \le u \]
		\item A $u$-ranking supermartingale $r$ with respect to $F^{\max}$:
		\[ \interpret{\nonterm{F^{\max}}}(r) + u \le r \]
		\item An invariant $\eta$ of $F$:
		\[ \interpret{F}(\eta) \ge \eta \]
	\end{itemize}
	where $F^{\max}$ is obtained from $F$ by replacing each $\min$ with $\max$ (its interpretation $\interpret{F^{\max}}$ is defined as expected), and $\nonterm{F}$ is defined inductively as follows.
	\begin{align*}
		&\nonterm{(X(\widetilde{e}))} &&= X(\widetilde{e}) \\
		&\nonterm{t} &&= 0 \\
		&\nonterm{(F_1 + F_2)} &&= \nonterm{F_1} + \nonterm{F_2} \\
		&\nonterm{(t \cdot F)} &&= t \cdot \nonterm{F} \\
		&\nonterm{(\ifexpr{\varphi}{F_1}{F_2})} &&= \ifexpr{\varphi}{\nonterm{F_1}}{\nonterm{F_2}} \\
		&\nonterm{(\min \{ F_1, \dots, F_n \})} &&= \min \{ \nonterm{F_1}, \dots, \nonterm{F_n} \}
	\end{align*}
\end{theorem}
\begin{proof}
	We first rewrite $F$ into the form $\min \{ F_1, \dots, F_n \}$ where $F_1, \dots, F_n$ do not contain $\min$.
	By Proposition~\ref{prop:min-affinise} and Theorem~\ref{thm:unbounded-reasoning-principle}, we can give a lower bound if there exists a common lower-bound certificate for any $t_1, \dots, t_n : D \to [0, 1]$ such that $t_1 + \dots + t_n = \probone$.
	Concretely:
	\begin{itemize}
		\item A prefixed point $u$ such that
		$\sum_{i = 1}^n t_i \cdot \interpret{F_i}(u) \le u$.
		\item A $u$-ranking supermartingale $r$ such that
		$\sum_{i = 1}^n t_i \cdot \interpret{\nonterm{F_i}}(r) + u \le r$.
		\item An invariant $\eta$ such that $\eta \le u$ and
		$\sum_{i = 1}^n t_i \cdot \interpret{F_i}(\eta) \ge \eta$.
	\end{itemize}

	We can simplify this condition.
	Actually, it is sufficient to consider only the case where $t_i = \probone$ for some $i$, and the above conditions are equivalent to the following.
	\begin{itemize}
		\item A prefixed point $u$ such that
		$\max_i \interpret{F_i}(u) \le u$.
		\item A $u$-ranking supermartingale $r$ such that
		$\max_i \interpret{\nonterm{F_i}}(r) + u \le r$.
		\item An invariant $\eta$ such that $\eta \le u$ and
		$\min_i \interpret{F_i}(\eta) \ge \eta$.
	\end{itemize}

	Since $\max$ also distributes over other operators similarly to $\min$, considering $\max_i F_i$ is equivalent to replacing \( \min \) with \( \max \) in the original quantitative formula.
\end{proof}

Thus, we can generalize Theorem~\ref{thm:unbounded-reasoning-principle} as follows.
\begin{corollary}
	Let $E$ be a fixed-point equation system with \(\min\).
	We have \( \eta' \le \mu \interpret{E} \) if there exist:
	\begin{enumerate}
		\item[$(a)$] a fixed-point equation system \( E' \) such that \( \interpret{E'} \le \interpret{E} \),
		\item[$(b)$] a prefixed point \( u \) of\/ \( \interpret{E'^{\max}} \), i.e., \( \interpret{E'^{\max}}(u) \le u \),
		\item[$(c)$] a \( u \)-ranking supermartingale \( r : D \to [0, \infty) \) of\/ \( \interpret{\nonterm{E'^{\max}}} \), and
		\item[$(d)$] an invariant \( \eta' \in \mathbb{E}_{\le u}(D) \) of\/ \( \interpret{E'} \), i.e., \( \eta' \le \interpret{E'}(\eta') \) with \( \eta' \le u \).
	\end{enumerate}
	where \( E'^{\max} \) is the equation system obtained by replacing every occurrence of \(\min\) in \( E' \) with \(\max\).
	\qed
\end{corollary}

\subsection{Max}
For the $\max$ operator, we cannot give a similar equation to Proposition~\ref{prop:min-affinise} because taking the dual of Proposition~\ref{prop:min-affinise} gives an equation for the \emph{greatest} fixed point, not the least fixed point.
Below, we sketch a naive approach to give a lower bound for the equation system with the $\max$ operator.

\[ X \quad=_{\mu}\quad \max \{ F_1, F_2 \} \]
We assume that $F_1$ and $F_2$ do not contain the $\max$ operator (nor the $\min$ operator).

\begin{proposition}\label{prop:max-affinise}
	Let $K_1, K_2 : \mathbb{E}(D) \to \mathbb{E}(D)$ be affine functions and $K = \max \{ K_1, K_2 \}$.
	\[ \sup_{t : D \to [0, 1]} \mu X. (t \cdot K_1(X) + (1 - t) \cdot K_2(X)) \quad\le\quad \mu X. \max \{ K_1(X), K_2(X) \} \]
\end{proposition}

Note that $\ge$ does not follow from the dual of Proposition~\ref{prop:min-affinise} because the dual is the following equation for the \emph{greatest} fixed point.
\[ \sup_{t : D \to [0, 1]} \nu X. (t \cdot K_1(X) + (1 - t) \cdot K_2(X)) \quad=\quad \nu X. \max \{ K_1(X), K_2(X) \} \]

By Proposition~\ref{prop:max-affinise}, we can give a lower bound if there exists a lower-bound certificate for \emph{some} $t : D \to [0, 1]$.
Concretely, a certificate consists of the following data.
\begin{itemize}
	\item A function $t : D \to [0, 1]$.
	\item A prefixed point $u$:
	\[ t \cdot \interpret{F_1}(u) + (1 - t) \cdot \interpret{F_2}(u) \le u \]
	\item A $u$-ranking supermartingale $r$:
	\[ t \cdot \interpret{F_1}(r) + (1 - t) \cdot \interpret{F_2}(r) + u \le r \]
	\item An invariant postfixed point $\eta$ such that $\eta \le u$:
	\[ t \cdot \interpret{F_1}(\eta) + (1 - t) \cdot \interpret{F_2}(\eta) \ge l \]
\end{itemize}

If we aim to implement a lower-bound verification algorithm for the $\max$ operator, we need to find good $t$.
However, we leave the problem of finding $t$ as an open problem.

\section{Supporting Integration}\label{sec:integration}

We consider allowing the integral operator in quantitative formulas.
In this case, we assume that a data type $\mathcal{D}$ is a set of measurable spaces.
We also restrict the set $\mathbb{E}(D)$ of expectations to the set of measurable functions $f : D \to [0, 1]$.
The set $\mathbb{E}(D)$ is an $\omega$cpo and an $\omega$ccpo.

\begin{definition}
	Let $(D_1, \Sigma_1)$ and $(D_2, \Sigma_2)$ be measurable spaces.
	A \emph{transition kernel} from $(D_1, \Sigma_1)$ to $(D_2, \Sigma_2)$ is a function $\kappa : D_1 \times \Sigma_2 \to [0, \infty]$ such that the following conditions hold.
	\begin{itemize}
		\item For any $x \in D_1$, $\kappa(x, {-})$ is a measurable on $(D_2, \Sigma_2)$.
		\item For any $A \in \Sigma_2$, the function $\kappa({-}, A) : D_1 \to [0, \infty]$ is a measurable function.
	\end{itemize}
\end{definition}

\begin{definition}
	We define the \emph{quantitative formulas} over quantitative predicate variables $\mathcal{X} = \{ X_1, \dots, X_n \}$ and data type $D$ as follows.
	\[ F \quad\coloneqq\quad \sum_{i = 1}^n \int X_i \,\mathrm{d} \kappa_i + t \]
	where $\kappa_i$ is a transition kernel from $(D, \Sigma)$ to $(D_i, \Sigma_i)$ and $t$ is a non-negative real term.
\end{definition}

Note that the constructors of quantitative formulas in Definition~\ref{def:fixed-point-equation} are subsumed in this definition.
The interpretation of a quantitative formula $F$ is given as follows.
\[ \interpret{\sum_{i = 1}^n \int X_i \,\mathrm{d} \kappa_i + t}(\eta)(v) \quad\coloneqq\quad \sum_{i = 1}^n \int \eta_i \,\mathrm{d} \kappa_i(v, {-}) + \interpret{t} \]

Now, we consider equation systems.
Similarly to Remark~\ref{rem:simplify-expectation}, we assume that equation systems consist of a single equation of the form
\[ E \quad=\quad \{\quad X(x) =_{\mu} \int X \,\mathrm{d} \kappa + t \quad\} \]
Suppose $u : D \to [0, \infty)$ is a prefixed point of $K = \interpret{E}$.
Then, $K_{\le u} : \mathbb{E}_{\le u}(D) \to \mathbb{E}_{\le u}(D)$ is Scott continuous by the monotone convergence theorem and Scott cocontinuous by the dominated convergence theorem (dominated by $u$).
It is obvious that $K_{\le u}$ is affine.
Therefore, Corollary~\ref{cor:unique-fixed-point-by-ranking-supermartingale} is also valid in this case.

\section{Preliminary Experiments for Negative Costs}\label{sec:experiment-negative-cost}

\begin{table}[bt]
	\caption{Experimental results for cost analysis with negative costs.}
	\label{tab:experiment-negative-cost}
	\begin{tabular}{c|cccc|cccc}
		& \multicolumn{4}{c}{Lower bound} & \multicolumn{4}{c}{Upper bound} \\
		Benchmark & Bound & Config & Result & Time (sec) & Bound & Config & Result & Time (sec) \\
		\hline
		\texttt{wang19\_fig3\_neg} & 247.5 & A & valid & 0.339 & 250.0 & A & valid & 0.262 \\
		\texttt{wang19\_fig3\_pos} & 100.0 & A & valid & 0.324 & 101.0 & A & valid & 0.244
	\end{tabular}
\end{table}

Although our results in this paper are based on the $\omega$cpo and $\omega$ccpo structures of $[0, \infty]$, our results can handle cost analysis with negative costs \cite{WangPLDI2019} by decomposing costs into positive and negative parts.
\[ c = c^{+} - c^{-} \qquad\text{where}\qquad c^{+} \coloneqq \max \{ c, 0 \}, \quad c^{-} \coloneqq \max \{ -c, 0 \}. \]

For example, consider the following coin flip program where $\tickstmt(\alpha)$ increases the cost by $\alpha \in \mathbb{R}$.
\[ \while{\mathtt{random\_bool}(1/2)}{\pbranch{\tickstmt(-1)}{1/3}{\tickstmt(1)}} \]
This program can be translated into the following pair of equations.
\begin{align*}
	X^{+}() \quad&=_{\mu}\quad \frac{1}{2} \left( \frac{1}{3} (0 + X^{+}()) + \frac{2}{3} (1 + X^{+}()) \right) + \frac{1}{2} \cdot 0 \\
	X^{-}() \quad&=_{\mu}\quad \frac{1}{2} \left( \frac{1}{3} (1 + X^{-}()) + \frac{2}{3} (0 + X^{-}()) \right) + \frac{1}{2} \cdot 0
\end{align*}
The solution of this system is $(X^{+}(), X^{-}()) = (2 / 3, 1 / 3)$, and thus the expected cost is $X^{+}() - X^{-}() = 1 / 3$.

We conducted preliminary experiments on cost analysis with negative costs using a benchmark from \cite{WangPLDI2019}, as shown in Table~\ref{tab:experiment-negative-cost}.

\ifdraft
\section{Uniqueness of Fixed-Points: \( 1 \)-Bounded Case}\label{sec:unique-fixed-points-1bounded}
As suggested in Introduction, our approach provides a lower bound \( \eta \) for \( \mu \interpret{E} \) based on two components:
a ranking argument establishing \( \mu \interpret{E} = \nu \interpret{E} \), and an invariant \( \eta \) of \( E \) (i.e., \( \eta \le \interpret{E}(\eta) \)).
This section formally describes the ranking argument and proves its correctness,
focusing on \( 1 \)-bounded \( E \).
A way to deal with general \( E \) is discussed in Section~\ref{sec:unique-fixed-points-unbounded}.

Section~\ref{sec:1ufp-informally} explain the basic idea of how to guarantee the uniqueness of fixed-points, and Section~\ref{sec:1ufp-formally} formalizes the argument and proves its correctness.

\subsection{Uniqueness of Fixed Points, Informally}\label{sec:1ufp-informally}

While our goal is to establish the coincidence of the least and greatest fixed points, let us instead consider a case (Example~\ref{ex:biased-random-walk}) where they do not coincide and analyze their difference.
As we have mentioned, the least fixed point of~\eqref{eq:biased-random-walk} gives the weakest preexpectation $\mathrm{wp}[c_{\mathrm{rw}}](f)$, which is the expected value of \( f \) after the execution.
The greatest fixed-point of \eqref{eq:biased-random-walk} is also a known notion: it is the \emph{weakest liberal preexpectation}, which is the expected value of $f$ plus the probability of non-termination.
So the difference between the greatest fixed point and the least fixed point in this case is the probability of non-termination.
Therefore, if the program $c_{\mathrm{rw}}$ were almost surely terminating, then the non-termination probability would be $0$, so
the fixed-point equation~\eqref{eq:biased-random-walk} should have a unique solution.

To assess the range of applicability of this approach, let us take a closer look at how the non-termination probability is related to the difference of the greatest fixed point and the least fixed point.
The least (greatest) solution is given as the least (greatest) fixed point of the function $K : \mathbb{E}_{\le 1}(\mathbb{Z}) \to \mathbb{E}_{\le 1}(\mathbb{Z})$ defined from the right-hand side of the equation~\eqref{eq:biased-random-walk}.
\[ K(X)(x) \quad\coloneqq\quad \ifexpr{x > 0}{\frac{1}{3} X(x - 1) + \frac{2}{3} X(x + 1)}{f(x)} \]
Since $K$ is Scott continuous and cocontinuous, \( \nu K = \inf_n K^n(\probone) \) and \( \mu K = \sup_n K^n(\probzero) \) (where $\probzero, \probone : \mathbb{Z} \to [0, 1]$ are the constant functions), so
\[ \nu K - \mu K \quad=\quad \inf_n (K^n(\probone) - K^n(\probzero)) \]
Although slightly abrupt, we now introduce a function \( \Diff K \colon \mathbb{E}_{\le 1}(\mathbb{Z}) \to \mathbb{E}_{\le 1}(\mathbb{Z}) \) defined by
\begin{equation*}
	(\nonterm{K})(X)(x) \quad\coloneqq\quad \ifexpr{x > 0}{\frac{1}{3} X(x - 1) + \frac{2}{3} X(x + 1)}{\probzero}
	\label{eq:biased-random-walk-conical}
\end{equation*}
where the right-hand-side is obtained by replacing \( f(x) \) in the definition \( K \) with \( \probzero \).
Intuitively, \( (\nonterm{K})^n(\probone) \) represents the probability that the program has not terminated after \( n \) steps.
This function \( \nonterm{K} \) has a conceptually simpler definition
\begin{equation}
	(\Diff K)(u) \quad=\quad K(u) - K(\probzero)
\end{equation}
and a calculation using this characterization of \( \Diff K \) and the affineness of \( K \) (Lemma~\ref{lem:affineness-of-interpretation}) shows
\begin{equation}
	K^n(\probone) - K^n(\probzero) \quad=\quad (\nonterm{K})^n(\probone).
	\label{eq:diff-gfp-lfp-1bounded}
\end{equation}
Therefore \( \mu K = \nu K \) if and only if the greatest fixed-point of \( \nonterm{K} \) is \( \probzero \).

Interestingly, the function $\nonterm{K}$ also characterize \emph{ranking supermartingales}~\cite{ChakarovCAV2013} for the program $c_{\mathrm{rw}}$, which is a function from program states to non-negative real numbers that decreases on average by a certain amount after each step of the program.
That is, a function $r : \mathbb{Z} \to [0, \infty)$ is a ranking supermartingale for the program $c_{\mathrm{rw}}$ if and only if $\nonterm{K}(r) + \probone \le r$.
As a ranking supermartingale is a witness of almost-sure termination, this observation justifies our criterion, namely, the greatest and least fixed-points coincide when the underlying program is almost-surely terminating.

\subsection{Uniqueness of Fixed Points, Formally}\label{sec:1ufp-formally}
This subsection formalizes the argument in Section~\ref{sec:1ufp-informally}.
The operator \( \nonterm{K} \) has a simple semantic characterization: \( \nonterm{K}(u) := K(u) - K(\probzero) \).
The following lemma relates this semantic characterization with a syntactic transformation.
\begin{lemma}\label{lem:equation-system-nonterminating-part}
	Given a quantitative formula $F$, we define $\nonterm{F}$ by the following syntactic translation: 
	\begin{gather}
		\nonterm{(X_i(\widetilde{e}))} \coloneqq X_i(\widetilde{e}) \qquad
		\nonterm{(F_1 + F_2)} \coloneqq \nonterm{(F_1)} + \nonterm{(F_2)} \qquad
		\nonterm{(t \cdot F)} \coloneqq t \cdot \nonterm{F} \\
		\nonterm{(t)} \coloneqq 0 \qquad
		\nonterm{(\ifexpr{\varphi}{F_1}{F_2})} \coloneqq \ifexpr{\varphi}{\nonterm{(F_1)}}{\nonterm{(F_2)}}
	\end{gather}
	Note the cases for \( t \cdot F \) and \( t \).
	Let $\nonterm{E}$ be the equation system obtained by applying the translation above to each equation in $E$.
	Then, we have $\nonterm{\interpret{E}} = \interpret{\nonterm{E}}$.
	\qed
\end{lemma}

We noted in Lemma~\ref{lem:affineness-of-interpretation} that the interpretation \( \interpret{E} \) of an equation system is affine.
Then its difference, \( \nonterm{\interpret{E}} \), has even better properties.

\begin{theorem}\label{thm:unique-fixed-point-1-bounded}
	Let \( K : \mathbb{E}_{\le 1}(D) \to \mathbb{E}_{\le 1}(D) \) be a function that is Scott continuous, Scott cocontinuous and affine.
	If there exists a ranking supermartingale \( r : D \to [0,\infty) \) of\/ \( \nonterm{K} \), then \( \inf_{n} \nonterm{K}^n(\probone) = \probzero \) and moreover, \( \mu K = \nu K \).
\end{theorem}
\begin{proof}
	Assume \( r \ge \Diff K(r) + \probone \).
	By the linearity of \( \nonterm{K} \) (Lemma~\ref{lem:affine-k-induces-linear-knt}),
	\begin{equation*}
		r
		\:\ge\:
		\nonterm{K}(r) + \probone
		\:\ge\:
		\nonterm{K}(\nonterm{K}(r) + \probone) + \probone
		\:=\:
		\nonterm{K}(\nonterm{K}(r)) + \nonterm{K}(\probone) + \probone
		\:\ge\cdots,
	\end{equation*}
	and hence \( r \ge (\nonterm{K})^{n+1}(r) + \sum_{k=0}^{n}(\nonterm{K})^k(\probone) \) for every \( n \).
	In particular, \( r(d) \in [0,\infty) \) is an upper bound of \( \sum_{k = 0}^n (\nonterm{K})^k(\probone)(d) \) for every \( n \) and \( d \in D \).
	Hence, for every \( d \in D \), the infinite sum \( \sum_{k=0}^{\infty} (\nonterm{K})^k(\probone)(d) \) (absolutely) converges to some \( v \le r(d) < \infty \).
	So the sequence \( ((\nonterm{K})^n(\probone)(d))_{n = 1}^{\infty} \) is summable, which implies \( \lim_{n \to \infty} (\nonterm{K})^n(\probone)(d) = 0 \).

	By Proposition~\ref{prop:diff-gfp-lfp-for-affine-k}, \( K^n(\probone) = (\nonterm{K})^n(\probone) + K^n(\probzero) \).
	By taking the limit of \( n \to \infty \),
	\begin{equation*}
		\lim_{n \to \infty} K^n(\probone) = \lim_{n \to \infty} (\nonterm{K})^n(\probone) + \lim_{n \to \infty} K^n(\probzero).
	\end{equation*}
	Since \( K \) is Scott continuous and Scott cocontinuous, we have \( \mu K = \lim_n K^n(\probzero) \) and \( \nu K = \lim_m K^m(\probone) \).
	Hence \( \nu K = 0 + \mu K \).
\end{proof}

\fi
\fi

\end{document}